\pdfoutput=1
\documentclass[10pt,letterpaper,pdftex,final,journal]{IEEEtran}

\usepackage{cite}
\usepackage[caption=false,font=footnotesize]{subfig}
\usepackage[final]{graphicx}
\usepackage{amsmath}
\usepackage{amsfonts}
\usepackage{amssymb}
\usepackage[amsmath,thmmarks]{ntheorem}
\usepackage{color}
\usepackage{bm}
\usepackage{import}
\usepackage{transparent}

\newcounter{MYtempeqncnt}

\title{
Asymptotic Analysis of SU-MIMO Channels With Transmitter Noise and Mismatched Joint Decoding}
\author{%
\IEEEauthorblockN{Mikko~Vehkaper{\"a}, Taneli~Riihonen, Maksym~Girnyk, 
Emil~Bj{\"o}rnson, \\
M{\'e}rouane~Debbah, Lars~K.~Rasmussen, and Risto~Wichman}
  \thanks{This research has been supported by the 
  Academy of Finland, the Swedish Research Council and the
  ERC Starting Grant 305123 MORE.}}

\newcommand{\varX}{X}

\newcommand{\varV}{V}

\newcommand{\setV}{\mathcal{\varV}}
\newcommand{\setX}{\mathcal{\varX}}
\newcommand{\funcF}{f}
\newcommand{\funcZ}{Z}

\newcommand{\outF}{f}

\DeclareMathOperator*{\extr}{extr}
\DeclareMathOperator*{\tr}{tr}
\DeclareMathOperator*{\diag}{diag}

\newcommand{\vm}[1]{\bm{#1}}
\newcommand{\E}{\mathsf{E}}

\newcommand{\dx}{\mathrm{d}}  
\newcommand{\im}{\mathrm{j}}
\newcommand{\e}{\mathrm{e}}
\newcommand{\Qp}{p}
\newcommand{\Qm}{m}
\newcommand{\Qq}{q}
\newcommand{\QQ}{Q}
\newcommand{\Qptil}{\tilde{\Qp}}
\newcommand{\Qqtil}{\tilde{\Qq}}
\newcommand{\QQtil}{\tilde{\QQ}}
\newcommand{\Qmtil}{\tilde{\Qm}}
\newcommand{\mtxQ}{\vm{Q}}
\newcommand{\mtxQtil}{\tilde{\mtxQ}}

\newcommand{\vsym}{v}
\newcommand{\vvec}{\vm{\vsym}}
\newcommand{\tvec}{\vm{d}}
\newcommand{\uvec}{\vm{t}}

\newcommand{\RmatP}{\tilde{\Rmat}}
\newcommand{\RmatSP}{\vm{\Sigma}}
\newcommand{\Omat}{\boldsymbol{\Omega}}
\newcommand{\OmatP}{\tilde{\Omat}}
\newcommand{\herm}{\mathsf{H}}
\newcommand{\trans}{\mathsf{T}}

\newcommand{\MGFu}{\phi^{(u)}}
\newcommand{\sigmaP}{\tilde{r}}
\newcommand{\xvec}{\vm{x}}
\newcommand{\xvecP}{\tilde{\xvec}}

\newcommand{\wvec}{\boldsymbol{w}}
\newcommand{\yvec}{\vm{y}}
\newcommand{\nsym}{n}
\newcommand{\nvec}{\vm{\nsym}}
\newcommand{\msym}{m}
\newcommand{\mvec}{\vm{\msym}}
\newcommand{\nusym}{\nu}
\newcommand{\nuvec}{\vm{\nusym}}
\newcommand{\osym}{\omega}
\newcommand{\ovec}{\vm{\osym}}
\newcommand{\Dvec}{\vm{\Delta}}

\newcommand{\chivec}{\vm{\chi}}
\newcommand{\Rmat}{\vm{R}}
\newcommand{\Gammamat}{\vm{\Gamma}}
\newcommand{\Imat}{\vm{I}}
\newcommand{\Hmat}{\vm{H}}
\newcommand{\gammabar}{\bar{\gamma}}
\newcommand{\xigamma}{\xi_{\gammabar}}
\newcommand{\sgamma}{s_{\gammabar}}
\newcommand{\sP}{\tilde{s}}

\theoremstyle{plain}
\newcommand{\hhnewtheorem}[2]{%
  \newtheorem{#2}{#1}
  \newenvironment{#1}[1][CCC]{%
    \begin{#2}[##1]
      \nopagebreak
      \begin{adjustwidth}{0.05\textwidth}{0.05\textwidth}}
      {\end{adjustwidth}
    \end{#2}}}
\newcommand{\hhnewproof}[3]{%
  \newtheorem{#3}{#2}
  \newenvironment{#1}{%
    \begin{#3}
      \nopagebreak
      \begin{adjustwidth}{0.05\textwidth}{0.05\textwidth}}
      {\end{adjustwidth}
    \end{#3}}}
\theoremseparator{.}
\hhnewtheorem{Theorem}{thm}
\hhnewtheorem{Proposition}{prop}
\hhnewtheorem{Claim}{claim}
\hhnewtheorem{Corollary}{cor}
\hhnewtheorem{Lemma}{lemma}
\theorembodyfont{\upshape}
\theoremsymbol{\ensuremath{\diamondsuit}}
\theoremseparator{.}
\hhnewtheorem{Example}{example}
\hhnewtheorem{Remark}{remark}
\hhnewtheorem{Definition}{defn}
\hhnewtheorem{Assumption}{assumption}
\hhnewtheorem{Step}{step}
\theoremstyle{nonumberbreak}
\theoremseparator{:}
\hhnewtheorem{Algorithm}{algorithm}
\theoremstyle{plain}
\theoremseparator{}
\theoremsymbol{}
\theoremstyle{nonumberplain}
\hhnewproof{Proof}{Proof:}{proof}

\newcommand{\openbox}{\leavevmode
  \hbox to.77778em{%
  \hfil\vrule
  \vbox to.675em{\hrule width.6em\vfil\hrule}%
  \vrule\hfil}}

\begin{document}

\maketitle

\begin{abstract}
	Hardware impairments in radio-frequency components of a wireless system
	cause unavoidable distortions to transmission that are not captured by
	the conventional linear channel model.  In this paper, a `binoisy'
	single-user multiple-input multiple-output (SU-MIMO) relation is considered where the 
	additional distortions are modeled via an additive noise term at the
	transmit side.
	Through this extended SU-MIMO channel model, the
	effects of transceiver hardware impairments on the achievable rate 
	of multi-antenna point-to-point systems are studied.  
	Channel input distributions encompassing 
	practical discrete modulation schemes, such as, QAM and PSK, 
	as well as Gaussian signaling 
	are covered.  In addition, the impact of mismatched detection 
	and decoding when the receiver has insufficient information 
	about the non-idealities is investigated. The numerical results 
	show that for realistic system parameters,
	the effects of transmit-side noise and mismatched decoding 
	become significant only at high modulation orders.
\end{abstract}

\section{Introduction}
\label{sec:intro}

\IEEEPARstart{M}{IMO}, i.e., multiple-input multiple-output, wireless links are a mature research subject and their theory is already well understood \cite{Tse-Viswanath-2005}. However, the extensive body of literature on link-level analysis conventionally concerns signal models of the form $\yvec = \Hmat \xvec + \nvec$ reckoning with an additive thermal-noise term, namely $\nvec$, only at the receiver after the fading channel $\Hmat$. In this paper, we investigate single-user MIMO channels and adopt a generalized (`binoisy') input--output relation from
\cite{Schenk-Smulders-Fledderus-SPS-DARTS2005,Goransson-Grant-Larsson-Feng-SPAWC2008,Suzuki-AnhTran-Collings-Daniels-Hedley-2008Aug,Suzuki-Collings-Hedley-Daniels-PIMRC2009,Studer-Wenk-Burg-WSA2010,GonzalezComa-Castro-Castedo-WSA2011,Studer-Wenk-Burg-EuCAP2011,GonzalezComa-Castro-Castedo-EW2011,Bjornson-Zetterberg-Bengtsson-Ottersten-2013Jan,Zhang-Matthaiou-Bjornson-Coldrey-Debbah-ICC2014}:
\begin{IEEEeqnarray}{l}
\label{eq:yvec_true}
\yvec = \Hmat (\xvec + \vvec) + \wvec,
\end{IEEEeqnarray}
where $\wvec$ is an additive receive-side distortion-plus-noise component.
The system model \eqref{eq:yvec_true} allows including an additive noise term, namely $\vvec$, also at the transmitter, thus making the total effective noise term $\Hmat\vvec + \wvec$ colored and correlated with the fading channel. This small but significant complement yields a MIMO link model whose performance analysis is still an open research niche in many respects. 

Although we primarily aim at extending the capacity theory of binoisy SU-MIMO channels under fading without committing to any particular application, 
the signal model \eqref{eq:yvec_true} originally stems from the practical need for modeling the combined effect of various transceiver hardware impairments which are detailed in~\cite{Fettweis-Lohning-Petrovic-Windisch-Zillmann-Rave-2007Jun,Schenk-Book2008}, and the references therein. However, it is worth acknowledging that the additive noise assumed herein is only a simplified representation of complex nonlinear phenomena occurring due to hardware impairments, especially when considering their joint coupled effects or trying to model residual distortion after compensation. Thus, the binoisy signal model should be regarded as a compromise between facilitating theoretical analysis and resorting to measurements or simulations under more accurate modeling. Yet the central limit theorem further justifies the model by averaging the combined effects of different impairments to additive Gaussian noise when the signal model \eqref{eq:yvec_true} is understood to represent a single narrowband subcarrier within a wideband system.

\begin{figure}[tb]
\centering
\includegraphics[width=\columnwidth]{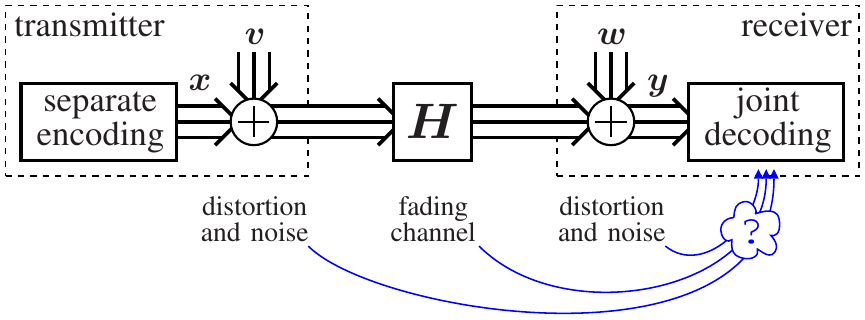}
\caption{System model for non-ideal MIMO communications with 
transmit and receive distortion.  The receiver might be 
misinformed or ignorant of some of the variables in the 
transmission chain leading to mismatched decoding.}
\label{fig:sysmodel}
\end{figure}

Additive receiver hardware impairments can be incorporated into the conventional signal model by increasing the level of the thermal-noise term $\nvec$ by a constant noise figure, e.g., about $3$--$5$~dB,
or by scaling it in proportion to the input signal level such that it matches with $\wvec$. On the other hand, regarding the joint effect of transmitter hardware impairments as an additive transmit-side noise term $\vvec$ is analogous to the principles of practical radio conformance testing. In particular, the common transmitter quality indicator is error-vector magnitude (EVM) which reduces the distortion effects to an additive component and measures its level relatively to signal amplitude~\cite{Gregorio-Cousseau-Werner-Riihonen-Wichman-2013Sep}.

Typical target EVM values guarantee that the signal $\xvec$ is at least $20$--$30$~dB
above the transmit-side noise $\vvec$.  On the other hand, for basic discrete
channel inputs such as quadrature phase-shift keying (QPSK), $\Hmat\xvec$ is usually 
at most $10$--$15$~dB above the receive-side noise $\wvec$, after which the
communication is not anymore limited by noise but the lack of entropy in the
modulation alphabet. This implies that transmitter hardware impairments can be
justifiably omitted in the analysis of simple low-rate wireless systems: Either
$\Hmat\vvec$ is well below the receive-side noise $\wvec$ (say $5$--$20$ dB)
or the signal-to-noise ratio (SNR) is set to an uninterestingly high level.  However,
there has been a trend to improve data rates by using, e.g., quadrature
amplitude modulation (QAM) up to 64-QAM at relatively high SNR, in which case the
transmit-side noise begins to play a notable role in the link-level performance.

The considered system setup corresponding to (\ref{eq:yvec_true}) is shown in Fig.~\ref{fig:sysmodel}. As for MIMO processing, we focus on regular spatial multiplexing where a conventional transmitter separately encodes and sends an independent stream at each of its antennas without having channel state information or being aware of the transmit-side noise it produces; the receiver jointly decodes the output signals of the MIMO channel knowing its instantaneous realization $\Hmat$ and some noise statistics. However, conventional receivers are designed and implemented based on the conventional signal model (where $\vvec=\vm{0}$) due to which they are prone to lapse into suboptimal {\em mismatched decoding} by inaccurately postulating the statistics of the actual noise term $\Hmat\vvec + \wvec$. Even if off-the-shelf receivers can adapt to colored receiver noise, they may not be able to track the variable statistics of the component $\Hmat\vvec$ propagated from the transmitter since it is correlated with the fading channel. Only an advanced receiver would be able to perform {\em matched decoding} knowing perfectly the noise statistics as if it was designed and implemented explicitly based on the generalized binoisy signal model (\ref{eq:yvec_true}).

\subsection{Related Works}

The key reference results for the present study are reported in
\cite{Schenk-Smulders-Fledderus-SPS-DARTS2005,Goransson-Grant-Larsson-Feng-SPAWC2008,Suzuki-AnhTran-Collings-Daniels-Hedley-2008Aug,Suzuki-Collings-Hedley-Daniels-PIMRC2009,Studer-Wenk-Burg-WSA2010,GonzalezComa-Castro-Castedo-WSA2011,Studer-Wenk-Burg-EuCAP2011,GonzalezComa-Castro-Castedo-EW2011,Bjornson-Zetterberg-Bengtsson-Ottersten-2013Jan,Zhang-Matthaiou-Bjornson-Coldrey-Debbah-ICC2014}.
These seminal works originally formulated the research niche around (\ref{eq:yvec_true}) and established the baseline understanding of MIMO communication in the presence of transmit-side noise with numerical simulations and theoretical analysis. The majority of the related works, e.g.,
\cite{Schenk-Smulders-Fledderus-SPS-DARTS2005,Goransson-Grant-Larsson-Feng-SPAWC2008,Studer-Wenk-Burg-WSA2010,Studer-Wenk-Burg-EuCAP2011}, concern regular spatial multiplexing using separate encoding like the present paper but also different variations of joint encoding have been creditably investigated, e.g., in~\cite{Suzuki-AnhTran-Collings-Daniels-Hedley-2008Aug,GonzalezComa-Castro-Castedo-WSA2011}. On the other hand, all the studies that we are aware of assume (implicitly) advanced receivers that know the presence of transmit noise, no matter what form of decoding is used.

Especially, the reference results are polarized such that the scope of analytical studies~\cite{Studer-Wenk-Burg-WSA2010,Studer-Wenk-Burg-EuCAP2011} typically differs from that of studies reporting simulations~\cite{Studer-Wenk-Burg-WSA2010,GonzalezComa-Castro-Castedo-WSA2011,GonzalezComa-Castro-Castedo-EW2011} or measurements~\cite{Suzuki-AnhTran-Collings-Daniels-Hedley-2008Aug,Suzuki-Collings-Hedley-Daniels-PIMRC2009,Studer-Wenk-Burg-WSA2010}. Except for~\cite{Schenk-Smulders-Fledderus-SPS-DARTS2005}, practical discrete modulation schemes, e.g., QAM, have not been previously analytically evaluated in the presence of transmit noise, and simulation-based studies usually concern bit/symbol/packet error rates, not transmission rates which could be more interesting when studying modern adaptive encoding. In contrast, all the analytical capacity studies assume Gaussian signaling and the throughput simulations of~\cite{Goransson-Grant-Larsson-Feng-SPAWC2008} with adaptive modulation and coding are their closest counterpart when it comes to experimental work.

If the receiver does not properly account for the additional transmit-side noise in the received signal, conventional mutual information (MI) is not anymore the correct upper bound for coded transmissions. Rather, due to mismatched decoding, one has to employ other metrics, such as {\em generalized mutual information} (GMI)~\cite{Merhav-etal-1994,Ganti-Lapidoth-Telatar-2000} adopted herein. Another common use for GMI is the analysis of bit-interleaved coded modulation~\cite{BICM-FnT-2008}, while also transceiver hardware impairments \cite{Zhang-2012Feb} and effects of imperfect channel state information at multi-antenna receiver \cite{Asyhari-Fabregas-2012, Weingarten-Steinberg-Shamai-2004} are analyzed in terms of GMI. In particular, MI and GMI are evaluated herein using the \emph{replica method} \cite{Nishimori-2001, Mezard-Montanari-2009}, originating from the field of statistical physics and introduced to the analysis of wireless systems by~\cite{Tanaka-2002Nov,Guo-Verdu-2005Jun}. Since then, the replica method has been applied to various problems in communication theory, e.g., MIMO systems~\cite{Muller-2003Nov,Wen-Wong-Chen-TCOM2007,Takeuchi-Muller-Vehkapera-Tanaka-2013}. For some special cases like Gaussian signaling, the replica trick renders exact asymptotic results when the number of antennas grows without bound, while they can be otherwise considered accurate approximations as shown by comparisons to Monte Carlo simulations. 

\subsection{Summary of Contributions}

In this paper, we investigate two aspects of binoisy MIMO channels that are unexplored in related works despite their fundamental role in understanding the effects of hardware impairments in wireless systems. 
Firstly, analytical capacity results are limited to Gaussian signaling while practical digital modulation is evaluated only based on simple simulations or measurements.
Secondly, the earlier literature focuses on the optimistic case of matched decoding by employing receivers that are actually not available off the shelf but implicitly updated to take account of transmit-side noise.

In particular, this paper contributes to the capacity theory of MIMO communication links by examining the effects of transmit-side noise as follows.
\begin{itemize}

\item Analytical GMI expressions are calculated for studying the rate loss of mismatched decoding when using a conventional receiver which is unaware of the transmit-side noise. Especially, it is shown that the performance remains the same irrespective of how well the noise covariance matrix is known if it is a constant.

\item The above analysis is further translated into corresponding asymptotic high-SNR limits for Gaussian signaling as a complement for the results of~\cite{Bjornson-Zetterberg-Bengtsson-Ottersten-2013Jan}, which covers matched decoding and conventional MI.

\item The analytical expressions provided for both conventional MI and GMI cover many practical discrete modulation schemes such as variations of PSK and QAM. 
This resolves the serious problem that evaluating (G)MI with direct Monte Carlo simulations 
for the present system is computationally infeasible except for cases with small number 
of antennas and low order modulation sets.

\end{itemize}
Extending beyond the scope of the paper, the replica analysis of GMI is also a new aspect at large.

\subsection{Outline of the Paper and Its Nomenclature}

After the considered system model is specified in the following section, the main analytical content of this paper is divided into two parts: Section~III concerns the performance of conventional suboptimal receivers under mismatched decoding, which is analyzed based on GMI; and Section~IV studies conventional MI with advanced receivers, which are aware of transmitter noise and, thus, capable of optimal matched decoding. In Section~V, the presented theory is illustrated with numerical results, including simulations for double-checking its accuracy, which is finally followed by concluding remarks in the last section.  Some general results from literature that are used throughout the paper for derivations are collected in Appendix~A for the convenience of the reader.  Appendices~B contains general description of the replica method and Appendix~C sketches the derivation of the main results in Section~III.

\emph{Notation:}
Complex Gaussian random variables (RVs) are always assumed to be 
proper and the density of such $\vm{x}\in\mathbb{C}^{N}$
with mean $\boldsymbol{\mu}$ and covariance $\Rmat$
is denoted $g(\vm{x} \mid \boldsymbol{\mu};\, \Rmat)$.
For the zero-mean proper Gaussians, we 
say they are circularly symmetric complex Gaussian (CSCG).
For convenience, both discrete and continuous 
RVs are said to have a probability density function (PDF) that 
is denoted by $p$, and we do not separate 
RVs and their realizations.  For postulated 
PDFs we write $q$ and add tilde on top of the related RVs 
(most of the time). Given a RV $x$ that has a PDF $p(x)$, we 
write $x\sim p(x)$ 
(and $\tilde{x}\sim q(\tilde{x})$ for the postulated 
case). 
Statistical expectation is denoted 
$\E\{\,\cdot\, \}$ and, unless stated otherwise, 
calculated over all randomness in the argument using true or 
postulated PDFs, depending on which type of RVs are present.
Integrals w.r.t.\ real-valued 
variables are always over $\mathbb{R}$ 
(for vectors over the appropriate product space) and we tend
to omit the integration limits for notational simplicity.
For a complex variable $z = x + \im y$, we denote
$\int ( \;) \dx z = \int ( \;) \dx x \dx y,$
and similarly for complex vectors.
Logarithms are natural logs and denoted $\ln$ unless 
stated otherwise.

\section{System model}
\label{sec:system_model}

Consider the system model depicted in Fig.~\ref{fig:sysmodel}
and the signal model of $\yvec \in \mathbb{C}^{N}$
written in (\ref{eq:yvec_true})
where $\Hmat \in \mathbb{C}^{N \times M}$ is the channel matrix
and $\xvec \in \mathbb{C}^{M}$ the signal of interest. 
The receive-side distortion plus noise component is 
divided into two parts, namely $\wvec = \nvec + \ovec \in \mathbb{C}^{N}$
where $\nvec$ is caused by thermal noise and $\ovec$ represents 
hardware impairments arising from the non-ideal 
behavior of the radio-frequency (RF) transceivers.  
Similarly, $\vvec = \mvec + \nuvec  \in \mathbb{C}^{M}$ where 
$\mvec$ and $\nuvec$ are related to thermal noise and 
hardware impairments or distortions, respectively, at the transmit-side.
In practice, the effect of $\mvec$ is often negligible compared to $\nuvec$.
In conventional MIMO literature it is common to consider 
only the thermal noise at the receiver, which translates to assuming
$\ovec = \nuvec = \mvec = \vm{0}$ in our more generic system model.

Let us denote the PDF
of the transmit vector $\xvec$ by
$p(\xvec)$ and assume it factorizes as
\begin{equation}
	\label{eq:p_of_x}
p(\xvec)  = \prod_{m=1}^{M} p(x_{m}),
\end{equation}
so that independent streams are transmitted at each transmit antenna.
Furthermore, let $p(x_{m})$ be a zero-mean distribution
with variance $\gammabar_{m}$.  
For later convenience, we let $\Gammamat$
be a diagonal matrix whose non-zero elements are
given by $\gammabar_{1}, \ldots, \gammabar_{M}$, that is,
$\Gammamat = \E \{\xvec \xvec^{\herm}\}$.
The channel $\Hmat$ is assumed to have independent identically distributed 
(IID) CSCG elements with variance%
\footnote{Typically the total power emitted from the transmit antennas 
in MIMO systems is constant; that is, $\tr(\Gammamat) = \gammabar$, where $\gammabar$ 
is some fixed power budget that does not depend on 
$M$.  Hence the elements of $\Gammamat$ need to be 
functions of $M$ in order to satisfy the 
transmit power normalization.
For the following analysis, however, it is more convenient 
to treat the elements of $\Gammamat$ to be independent
of $M$ and let the transmit power normalization be a part of the 
channel.  Clearly, both approaches are mathematically 
fully equivalent.} $1/M$.
The thermal noise samples at the transceivers 
are modeled as CSCG random vectors 
$\mvec$ and $\nvec$ that have independent elements.
For simplicity, we assume that any given noise or hardware
impairment component
is independent of any other RVs in the system. 
The transmit- and receive-side
impairments $\nuvec$ and $\ovec$ are taken to be CSCG
random vectors with covariance matrices
$\Rmat_{\nuvec}$ and $\Rmat_{\ovec}$, respectively.
The distortion plus noise vectors 
$\vvec$ and $\wvec$ are thus CSCG random vectors
whose covariance matrices we denote $\Rmat_{\vvec}$ and $\Rmat_{\wvec}$,
respectively.
Notice that these matrices can be functions
of the statistics of some other RVs albeit
we suppress the explicit statement of such dependence 
at this point for notational convenience.
The SNR without transmit-side noise
is defined as $\tr(\Gammamat)/\tr(\Rmat_{\wvec})$.

The PDF of the received signal, 
conditioned on $\xvec$, $\vvec$ and $\Hmat$, is given by
\begin{equation}
p(\yvec \mid \xvec, \vvec, \Hmat) = 
g(\yvec \mid \Hmat (\xvec+\vvec);\, \Rmat_{\wvec}),
\label{eq:true_pdf_cond_all}
\end{equation}
and the receiver is assumed to know $\Hmat$ and the true 
distribution $p(\xvec)$ of the channel input.
However, the additional transmit-side 
term $\vvec$ is in general  unknown at the receive-side 
and, thus, the PDF
\eqref{eq:true_pdf_cond_all} cannot be directly used for 
detection and decoding. 
Herein, we consider two different scenarios for the 
joint decoding operation at the receiver:
\begin{enumerate}
\item The receiver knows $\Hmat$, the PDFs of the
noise plus distortion terms $\vvec$ and $\wvec$
as well as the distribution of the  data vector $\xvec$.  
\emph{Matched joint decoding} is then based on the
conditional PDF
\begin{IEEEeqnarray}{rCl}
\label{eq:true_decoding_pdf}
p(\yvec \mid \xvec, \Hmat) &=& 
\E_{\vvec}\{
g(\yvec \mid \Hmat (\xvec+\vvec);\, \Rmat_{\wvec})\} 
\IEEEeqnarraynumspace\\
&=&
g(\yvec \mid \Hmat \xvec;\, \Rmat_{\wvec}+\Hmat \Rmat_{\vvec} \Hmat^{\herm}),
\IEEEeqnarraynumspace
\label{eq:true_decoding_pdf_expanded}
\end{IEEEeqnarray}
where the second equality follows by first 
using \eqref{eq:Gint} to calculate the expectation w.r.t.\ 
$\vvec$ and simplifying the end result 
using \eqref{eq:matrix_id1} and \eqref{eq:det_id1}.
Note that the effective noise covariance matrix in 
\eqref{eq:true_decoding_pdf_expanded} depends now on the 
instantaneous channel realization $\Hmat$.
\item The receiver has perfect knowledge of $\Hmat$ and the PDF
of the data vector $\xvec$. Instead of \eqref{eq:true_decoding_pdf}, 
however, the device uses a postulated channel law
\begin{equation}
\label{eq:post_decoding_pdf}
q(\yvec \mid \xvec, \Hmat) = 
 g(\yvec \mid \Hmat \xvec;\, \RmatP),
\end{equation}
for \emph{mismatched joint decoding} 
\cite{Merhav-etal-1994,Ganti-Lapidoth-Telatar-2000}.  
In contrast to $\Rmat_{\wvec}$ in \eqref{eq:true_decoding_pdf_expanded}, that is a random
matrix, the postulated covariance matrix $\RmatP$ in \eqref{eq:post_decoding_pdf} is fixed.
\end{enumerate}

If matched joint decoding is employed, the conventional metric
for evaluating the (ergodic) achievable rate of the system 
for given input distribution $p(\xvec)$ is the
MI between the channel inputs and outputs,
namely,
\begin{equation}
\label{eq:MI}
I(\yvec;\, \xvec)
= \E \{\ln p(\yvec \mid \xvec, \Hmat)\} - 
\E  \{\ln p(\yvec \mid \Hmat)\},
\end{equation}
where $p(\yvec \mid \Hmat) = \E_{\xvec}\{p(\yvec \mid \xvec, \Hmat)\}$
and the expectation is w.r.t.\ all RVs in the system model,
including the channel matrix $\Hmat$.
From the system design perspective, however, it might be impractical to use
\eqref{eq:true_decoding_pdf_expanded} due to complexity of 
implementation, resulting in mismatched decoding.  
To lower bound the true maximum rate that can be achieved 
reliably over channel \eqref{eq:yvec_true} when 
decoding rule \eqref{eq:post_decoding_pdf} is used at the receiver,
we use GMI that is discussed in the next section.

\section{Mismatched Joint Decoding: \\
	Generalized Mutual Information}

\label{sec:GMI}

\subsection{Definition and the Special Case of Gaussian Signaling}

Let us assume that the received signal is given by \eqref{eq:yvec_true}
but the receiver uses \eqref{eq:post_decoding_pdf} for 
decoding.  
Given $p(\xvec)$, the (ergodic) GMI
between the channel inputs and outputs 
is defined as~\cite{Merhav-etal-1994,Ganti-Lapidoth-Telatar-2000}
\begin{equation}
\label{eq:Igmi}
I_{\mathsf{GMI}}(\yvec;\, \xvec)
= \sup_{s>0} I^{(s)}_{\mathsf{GMI}}(\yvec;\, \xvec),
\end{equation}
where, denoting $q^{(s)}(\yvec \mid \Hmat) 
= \E_{\xvec} \{q(\yvec \mid \xvec, \Hmat)^{s}\},$
the $s$-dependent part reads
\begin{equation}
\label{eq:Igmi_S}
I^{(s)}_{\mathsf{GMI}}(\yvec;\, \xvec)
= \E \{\ln q(\yvec \mid \xvec, \Hmat)^{s}\} - 
\E  \{\ln q^{(s)}(\yvec \mid \Hmat)\}.
\end{equation}
Since we consider ergodic rates, the expectations in \eqref{eq:Igmi_S}
are w.r.t.\ all RVs in the system model, including the channel matrix $\Hmat$.
If $I$ is the maximum ergodic rate that can be transmitted 
over the channel \eqref{eq:yvec_true} using input distribution 
$p(\xvec)$ and decoding rule \eqref{eq:post_decoding_pdf}, 
then $I\geq I_{\mathsf{GMI}}$ \cite{Merhav-etal-1994,Ganti-Lapidoth-Telatar-2000}.  
Herein, the decoding based on 
the true channel law
\eqref{eq:true_decoding_pdf} cannot be obtained 
as a special case of the mismatched case
since $\RmatP$ is fixed (see footnote~\ref{fn:nonvalid_for_matched} and
\eqref{eq:eG_1} in Appendix~\ref{sec:replicas}) and, thus,
the case of matched decoding is considered separately in
Section~\ref{sec:matched_decoding}.

We are first interested in evaluating 
the $s$-dependent part of the normalized GMI per transmit stream
$M^{-1}I^{(s)}_{\mathsf{GMI}}(\yvec;\, \xvec)$ 
for given $s>0$.  The optimization over the free parameter 
$s$ is carried out after the suitable expressions are 
found. The first term in \eqref{eq:Igmi_S}
can be written as
\begin{IEEEeqnarray}{l}
\frac{1}{M}\E \{\ln q(\yvec \mid \xvec, \Hmat)^{s} \}
\IEEEeqnarraynumspace\IEEEnonumber\\ [-1ex]
= - \overbrace{ \frac{s}{M}\big[N \ln \pi
+ \ln \det \RmatP
\big]}^{=c^{(s)}} 
\IEEEnonumber\\
\qquad - \frac{s}{M}\E\big\{ (\Hmat\vvec + \wvec)^{\herm}
\RmatP^{-1}
(\Hmat\vvec + \wvec)\big\} 
\IEEEnonumber\\
 = - c^{(s)} 
 - \frac{s}{M}\bigg[\tr(\RmatP^{-1}\Rmat_{\wvec})
+
\frac{1}{M}\tr ( \RmatP^{-1})\tr(\Rmat_{\vvec})
\bigg].
\IEEEeqnarraynumspace 
 \label{eq:first_term_in_MI}
\end{IEEEeqnarray}
The first equality follows from
\eqref{eq:post_decoding_pdf}
by the fact that $\yvec - \Hmat \xvec = 
\Hmat\vvec + \wvec$ when $\xvec$ is given.  The second
equality is a consequence of the assumption that
the channels and noise vectors 
are all mutually independent and $\Hmat$ has zero-mean 
IID entries with variance $1/M$.
Notice that \eqref{eq:first_term_in_MI} is independent of 
$p(\xvec)$ and hence valid for all channel inputs.
Evaluating the second term in \eqref{eq:Igmi_S} is 
more complicated but
\emph{for the special case of Gaussian 
inputs} we have the result shown below.

\begin{example}
\label{example:Gaussian_simulation_mismatched}
For the special case of Gaussian 
inputs; that is, $p(\xvec) = g(\xvec \mid \vm{0};\, \Gammamat)$,
\begin{IEEEeqnarray}{l}
	\frac{1}{M}I^{(s)}_{\mathsf{GMI}}(\yvec;\, \xvec)
	=
	\frac{1}{M} \E_{\Hmat}\bigg\{  \ln \det\big(\RmatP + 
	s \Hmat \Gammamat \Hmat^{\herm}\big)
	\IEEEnonumber\\
	\;  + s 
	\tr\big[
	\big(\Rmat_{\wvec}
	+\Hmat(\Rmat_{\vsym}+ \Gammamat) \Hmat^{\herm}
	\big) \big(\RmatP + 
	s \Hmat \Gammamat \Hmat^{\herm}\big	)^{-1}\big]
	\IEEEnonumber\\
	\; 
-s\tr(\RmatP^{-1}\Rmat_{\wvec})
-\frac{s}{M}\tr ( \RmatP^{-1})\tr(\Rmat_{\vvec})
	- \ln \det \RmatP
	\bigg\}.
	\IEEEeqnarraynumspace
	\label{eq:Igmi_gauss}
\end{IEEEeqnarray}
The result is obtained by first using \eqref{eq:Gint}
and then simplifying with 
\eqref{eq:matrix_id1} and \eqref{eq:det_id1}.  Inserting 
the RHS of \eqref{eq:yvec_true} into the obtained expression and 
taking the expectations w.r.t.\ the noise terms
$\vvec$ and $\wvec$ completes the derivation.
\end{example}

Example~\ref{example:Gaussian_simulation_mismatched} shows that for Gaussian signals 
we only need to average over the channel to obtain the 
$s$-dependent part of GMI.  This is 
doable with Monte Carlo simulation. 
However, finding the optimal $s$ is 
time consuming even in this case and
a simple analytical expression that does not explicitly 
depend on the form of the marginals in \eqref{eq:p_of_x}
would be highly desirable.  
With this in mind, we adopt the following restriction to our system model 
from the physical characteristics of typical real transmitters for simplifying 
the analysis.

\begin{assumption}
\label{ass:uncorrelated_m}
The covariance matrix for the transmit-side
distortion plus noise term $\vvec$ is diagonal so that 
we may write $\Rmat_{\vvec}
= \Rmat_{\mvec} +  \Rmat_{\nuvec} = 
\diag(r^{(1)}_{\vvec}, \ldots, r^{(M)}_{\vvec})$.  Hence, $\vvec$ 
has independent (but not necessarily identically 
distributed) entries drawn according to 
$p(\vsym_{m}) = g (\vsym_{m} \mid 0 ;\,r^{(m)}_{\vvec})$. 
\end{assumption}

The physical meaning of this assumption is that hardware impairments 
at different transmitter branches arise in separate electrical components and 
there are no mechanisms which generate significant correlation between the 
elements of the distortion noise vector. Furthermore, it is actually not 
necessary for the replica analysis but it helps simplify the end result to a 
form whose numerical evaluation is computationally easy.

\subsection{Analytical Results via the Replica Method}
\label{sec:mismatched_results}

If the goal is to calculate the expectations related to
the latter term in \eqref{eq:Igmi_S} analytically and for general 
input distributions, we need 
to employ somewhat more advanced analytical tools than 
the basic probability calculus used in
Example~\ref{example:Gaussian_simulation_mismatched}.  
As we shall see shortly, employing the replica method
provides a formula that is applicable to a variety of input 
constellations, such as Gaussian or QAM.
To begin, let us first denote
\begin{IEEEeqnarray}{l}
-\frac{1}{M}
\E  \ln q^{(s)}(\yvec \mid \Hmat)
= 
c^{(s)} + f (s),  \IEEEeqnarraynumspace
\end{IEEEeqnarray}
where $c^{(s)}$ is defined in \eqref{eq:first_term_in_MI}
and the latter term, equivalent of the 
so-called \emph{free energy} in 
statistical mechanics, reads
\begin{IEEEeqnarray}{l}
f (s)
 \label{eq:mismatched_free_energy}
 \\
=
- \frac{1}{M}
\E \Big\{\! \ln \E_{\xvecP} \big\{\e^{-
[\Hmat (\xvec + \vvec - \xvecP) + \wvec]^{\herm}s\RmatP^{-1}
[\Hmat (\xvec + \vvec - \xvecP)+ \wvec]}\big\}\Big\}.
\IEEEnonumber
\end{IEEEeqnarray}
Now the inner expectation over the postulated channel 
input $\xvecP$ is w.r.t.\ a generic PDF 
\eqref{eq:p_of_x} and cannot be solved using 
\eqref{eq:Gint} as before.  %
The outer expectation is w.r.t.\ the rest of the 
RVs in the system, namely $\{\xvec, \vvec, \wvec, \Hmat\}$.
Due to \eqref{eq:Igmi_S} and \eqref{eq:first_term_in_MI}
the expression to be optimized in the GMI formula thus becomes
\begin{IEEEeqnarray}{l}
\frac{1}{M}
I^{(s)}_{\mathsf{GMI}}(\yvec;\, \xvec)	
\IEEEnonumber\\
\quad = f (s)
-
\frac{s}{M}\bigg[\tr(\RmatP^{-1}\Rmat_{\wvec})
+\frac{1}{M}\tr ( \RmatP^{-1})\tr(\Rmat_{\vvec})
\bigg]. \IEEEeqnarraynumspace
\label{eq:GMI_s_dependent_freeE}
\end{IEEEeqnarray}

\begin{remark}
\label{remark:scaled_identities}
By \eqref{eq:mismatched_free_energy} and \eqref{eq:GMI_s_dependent_freeE},
it is clear that if the receiver assumes that the additive noise in the 
system is spatially white $\RmatP = \sigmaP\Imat_{N}$ with some finite 
sample variance $\sigmaP$, the GMI remains the same for all $\sigmaP>0$ since 
the optimization over $s>0$ in \eqref{eq:Igmi} can be replaced by an optimization 
over a new variable $\sP = s / \sigmaP>0$.  Thus,
if the receiver uses $\RmatP = \sigmaP\Imat_{N}$
for decoding, the GMI is the same for all 
$\sigmaP > 0$ when the transmit and receive covariance matrices 
$\Rmat_{\vvec}$ and $\Rmat_{\wvec}$ are fixed. 
\end{remark}

The main obstacle in evaluating \eqref{eq:GMI_s_dependent_freeE} is 
clearly $f (s)$.  %
This term happens to be, however, of a form that can 
be tackled by the replica method (see Appendix~\ref{sec:replica_overview}). 
The following result is derived in Appendix~\ref{sec:replicas}
under the assumption of the so-called \emph{replica symmetric} (RS) 
ansatz when the system approaches the \emph{large system limit} (LSL),
that is, $M,N\to \infty$ with finite and 
fixed ratio $\alpha = M/N > 0$.
The limit notation is omitted below and 
the results should therefore be interpreted as 
approximations for systems that have finite dimensions.

\begin{prop}
\label{prop:mismatched_general}
Let $m=1,\ldots,M$ and denote
\begin{IEEEeqnarray}{rCl}
	\label{eq:mismatched_correct_input_decoupled}
\chi_{m} &=& x_{m} +  \vsym_{m}, \IEEEeqnarraynumspace\\
\tilde{\chi}_{m} &=& \tilde{x}_{m},
	\label{eq:mismatched_postulated_input_decoupled}
\end{IEEEeqnarray}
where $x_{m},\tilde{x}_{m} \sim p(x_{m})$ and 
$\vsym_{m} \sim g (\vsym_{m} \mid 0 ;\,r^{(m)}_{\vvec})$
are independent for all $m$ by assumption.
Let
\begin{IEEEeqnarray}{rCl}
p(z_{m} \mid \chi_{m}) &=& g(z_{m} \mid \chi_{m};\,\eta^{-1}), \\
q(z_{m} \mid \tilde{\chi}_{m}) &=& g(z_{m} \mid \tilde{\chi}_{m};\, \xi^{-1}),
\end{IEEEeqnarray}
be the PDF of an output $z_{m}$ of an
additive white Gaussian noise (AWGN) channel whose 
input is either \eqref{eq:mismatched_correct_input_decoupled}
or \eqref{eq:mismatched_postulated_input_decoupled}, respectively, 
and corrupted by additive noise with 
variance $\eta^{-1}$ or $\xi^{-1}$, respectively.
The parameters $\eta,\xi$ satisfy 
\begin{IEEEeqnarray}{rCl}
\label{eq:mismatched_eta_general}
\eta &=&
\frac{1}{\alpha}
\frac{\Big[ \frac{1}{N}
\tr \big(\OmatP^{-1}\big)\Big]^{2}}
{\frac{1}{N}
\tr \big(\OmatP^{-1}\Omat\OmatP^{-1}\big)}, \\
\xi &=& 
\frac{1}{\alpha N}
\tr \big(\OmatP^{-1}\big),
\label{eq:mismatched_xi_general}
\end{IEEEeqnarray}
for the given matrices
\begin{IEEEeqnarray}{rCl}
	\label{eq:Omat}
\Omat
&=& 
\Rmat_{\wvec} + \varepsilon \Imat_{N}, \IEEEeqnarraynumspace\\
\OmatP  &=& 
s^{-1}\RmatP+\tilde{\varepsilon}\Imat_{N},\IEEEeqnarraynumspace
\label{eq:OmatP}
\end{IEEEeqnarray}
and variables
\begin{IEEEeqnarray}{rCl}
\label{eq:mismatched_epsilon}
\varepsilon &=& 
\frac{1}{M} \sum_{m=1}^{M} 
\E \big\{| \vsym_{m} + x_{m} - \langle \tilde{x}_{m} \rangle_{q}|^{2}\big\},
\IEEEeqnarraynumspace\\
\tilde{\varepsilon} 
&=&
\frac{1}{M} \sum_{m=1}^{M} 
\E\big\{
| \tilde{x}_{m}-\langle \tilde{x}_{m} \rangle_{q}|^{2}\big\}.
\IEEEeqnarraynumspace
 \label{eq:mismatched_epsilon_postulated}
\end{IEEEeqnarray}
The notation $\langle \tilde{x}_{m} \rangle_{q}$ above 
refers to a decoupled posterior mean estimator
\begin{IEEEeqnarray}{rCl}
	\label{eq:mismatched_pme}
\langle \tilde{x}_{m} \rangle_{q} &=&  
\frac{\E_{\tilde{x}_{m}}\{\tilde{x}_{m} 
q(z_{m} \mid \tilde{x}_{m})\}}{q(z_{m})},
\IEEEeqnarraynumspace
\end{IEEEeqnarray}
where
$q(z_{m}) = 
\E_{\tilde{\chi}_{m}} \{q(z_{m} \mid \tilde{\chi}_{m})\}$.
If we also write 
$p(z_{m}) = 
\E_{\chi_{m}}\{ p(z_{m} \mid \chi_{m})\}$,
the free energy $f (s)$ defined in \eqref{eq:GMI_s_dependent_freeE}
is given under the assumption of the RS ansatz by
\begin{IEEEeqnarray}{rCl}
f_{\mathsf{RS}} (s) 
&=&
\frac{1}{\alpha N}\big[\ln \det\OmatP
+ \tr\big(\OmatP^{-1}\Omat\big)
- \ln \det (s^{-1}\RmatP)\big] 
\IEEEeqnarraynumspace\IEEEnonumber\\
&&-\bigg(\ln \frac{\pi}{\xi}
+ \frac{\xi}{\eta}
+ 
\frac{1}{M} \sum_{m=1}^{M} \int 
p(z_{m}) \ln q(z_{m}) \dx z_{m}\bigg) 
\IEEEnonumber\\
&&-\xi\varepsilon
 + \frac{\xi(\xi-\eta)}{\eta} \tilde{\varepsilon}.
\label{eq:FreeE_mismatched_general}
\end{IEEEeqnarray}
If multiple solutions to the coupled fixed point equations
\eqref{eq:mismatched_eta_general}~--~%
\eqref{eq:mismatched_epsilon_postulated}
are found, the one 
minimizing \eqref{eq:FreeE_mismatched_general} should be 
chosen.
\end{prop}

\begin{proof}
An outline of the derivation is given in
Appendix~\ref{sec:replicas}.
\end{proof}

The above result extends some previous works
such as \cite{Tanaka-2002Nov,Guo-Verdu-2005Jun}
in the direction of correlated noise at the receiver and 
additive transmit-side impairments.  
It is thus clear that the original GMI term \eqref{eq:Igmi_S}
of the MIMO system that suffers from transceiver 
hardware impairments has an interpretation in terms of an 
equivalent \emph{decoupled}%
\footnote{%
This decoupling property is ubiquitous in replica analysis
(see for example \cite{Tanaka-2002Nov,Guo-Verdu-2005Jun})
as well as in random matrix theory 
(see \cite{Tulino-Verdu-2004FnT,Couillet-Debbah-2011} and references therein),
and is one of the key reasons why the asymptotic methods provide
computationally feasible solutions for complex problems.}
scalar system.  This decoupled channel has only additive
distortions but unlike in the conventional case of replica analysis
\cite{Tanaka-2002Nov,Guo-Verdu-2005Jun}, the transmit-side 
has its own noise term. 
It should be remarked, however, that the implicit assumption here is that
$f_{\mathsf{RS}}(s) = f(s)$; that is, the system is not replica symmetry 
breaking (RSB).  We leave the RSB case as a possible future work and 
check the validity of the solution with selected numerical simulations.

For simplicity of presentation, we consider next a few 
practical special cases of Proposition~\ref{prop:mismatched_general} 
where the transmit power is the same for all antennas
and the noise and distortions at the transmit-side are spatially 
uncorrelated, namely,
$\Gammamat = \gammabar \Imat_{M}$ and $\Rmat_{\vvec} = r_{\vvec}\Imat_{M}$.
The receiver postulates spatially white noise 
$\RmatP = \sigmaP\Imat_{N}$ with some variance $\sigmaP>0$.
This 
allows us to write
\begin{IEEEeqnarray}{l}
	\label{eq:GMI_s_dependent_freeE_simple}
	\frac{1}{M} I_{\mathsf{GMI}}(\yvec;\, \xvec)
	=\sup_{\sP>0} 
	\Big\{
	f (\sP) - \alpha^{-1}\sP
  [N^{-1}\tr (\Rmat_{\wvec})
  + r_{\vvec}
  ] \Big\},
   \IEEEeqnarraynumspace
\end{IEEEeqnarray}
where $f (\sP)$ is given by \eqref{eq:mismatched_free_energy}
with $s\RmatP^{-1}$ replaced by $\sP\Imat_{N}$.
Furthermore, in this case all variables are identically distributed
for all $m=1,2,\ldots,M$ so we may omit the subscripts 
related to $m$ in the following.
We still need to fix the input distribution \eqref{eq:p_of_x} to obtain 
the parameters
\eqref{eq:mismatched_epsilon} and \eqref{eq:mismatched_epsilon_postulated}.
For this, we give two concrete 
examples: 1) Gaussian signaling; and 2)  
discrete channel inputs, such as, QAM.

\begin{example}
\label{example:mismatched_gaussian}
Let the channel inputs \eqref{eq:p_of_x} be IID Gaussian, namely,
$p(\xvec) = g(\xvec \mid \vm{0};\, \gammabar\Imat_{M})$ so that
$p(\tilde{\chi}_{m}) = p(x_{m}) = g(x \mid 0;\, \gammabar)$ and
$p(\chi_{m}) 
= g(\chi_{m} \mid 0; \, \gammabar + r_{\vvec})$ 
in Proposition~\ref{prop:mismatched_general}.
The parameter $\xi$ can then be obtained explicitly as
\begin{equation}
	\label{eq:xi_gaussian}
\xi =
\frac{\gammabar \sP (1- \alpha) - \alpha + 
\sqrt{4 \alpha \gammabar \sP  + [ \gammabar \sP (1- \alpha) - \alpha]^{2}}}
{2\alpha \gammabar},
\end{equation}
while $\eta$ and $\varepsilon$ are obtained by solving 
the coupled fixed point equations
\begin{IEEEeqnarray}{rCl}
\label{eq:eta_simple_Gauss}
\eta &=&
\frac{1}{\alpha[N^{-1}\tr(\Rmat_{\wvec}) + \varepsilon]}, \IEEEeqnarraynumspace\\
\label{eq:mse_gauss}
\varepsilon
&=&
\frac{\eta r_{\vvec}+\gammabar(\eta+\xi^{2} \gammabar)}{\eta(1+\xi \gammabar)^{2}}
=
\frac{\gammabar+r_{\vvec}}{(1+\xi \gammabar)^{2}}
+\frac{1}{\eta(1+1/\xi \gammabar)^{2}}.
\IEEEeqnarraynumspace 
\end{IEEEeqnarray}
Additional algebra shows that 
for IID Gaussian inputs, the free energy
\eqref{eq:FreeE_mismatched_general}
reduces to
\begin{IEEEeqnarray}{rCl}
f_{\mathsf{RS}} (\sP) \! &=& \! 
\frac{1}{\alpha}\bigg(
\frac{\xi}{\eta} 
+ \ln \sP +
\ln \frac{1}{\alpha \xi}
\bigg) \! -\xi\varepsilon
+\ln (1+\xi \gammabar)
\!+\! 
\frac{\xi r_{\vvec}}{1+\xi\gammabar}.
\IEEEnonumber\\
\label{eq:FreeE_mismatched_scaled_identities_gaussian}
\end{IEEEeqnarray}
Note that the expression for parameter 
$\tilde{\varepsilon}$ in
\eqref{eq:mismatched_epsilon_postulated}
is not explicitly given here but it is implicitly a part 
of \eqref{eq:xi_gaussian} due to relations 
\eqref{eq:mismatched_xi_general}~and~\eqref{eq:OmatP}.
\end{example}

The computational formula for obtaining the GMI with the above example is 
detailed in Table~\ref{table:prop1}.
\begin{table}[!t]
\renewcommand{\arraystretch}{1.2}
 \setlength{\tabcolsep}{3pt}
\caption{How to Obtain GMI for Gaussian Signaling from Example~\ref{example:mismatched_gaussian}}
\label{table:prop1}
\centering
\begin{tabular}{lp{77mm}}
\hline\\[-2mm]
	1) &
	Choose the parameters that define the MIMO system of interest, namely,
	antenna ratio $\alpha = M/N$, transmit- and receive-side 
	distortion plus noise covariance matrices 
	$\Rmat_{\vvec} = r_{\vvec}\Imat_{M}$ and
	$\Rmat_{\wvec}$, respectively, and the average transmit power 
	per antenna $\gammabar$.  Let also the optimization parameter
	$\sP>0$ be given. \\
	2) & 
	Plug the values of $\{\alpha,\gammabar,\sP\}$
	to \eqref{eq:xi_gaussian} and obtain $\xi$. \\
	3) &
	Insert $\xi$ along with the rest of the necessary parameters in 
	\eqref{eq:eta_simple_Gauss}~and~\eqref{eq:mse_gauss}, and solve 
	$\eta$ numerically, e.g., 
	using an iterative substitution method. \\
	4) & 
	Use the solutions of $\xi$ and $\eta$ in
	\eqref{eq:FreeE_mismatched_scaled_identities_gaussian}
	to obtain the free energy. \\
	5) & 
	Optimize 
	\eqref{eq:GMI_s_dependent_freeE_simple}
	over $\sP>0$. \\ [1mm]
	\hline
\end{tabular}
\end{table}
Notice that there are two non-trivial steps in the algorithm: 
1) the optimization over $s>0$; and 2) the problem
of solving a system of two nonlinear equations with two unknowns.
The first difficulty is not specific to the 
current study and is present in any work that considers GMI as means
to analyze mismatched decoding.  The computational complexity 
of the second problem is negligible compared to the original 
task of taking an expectation over the 
channel matrices in \eqref{eq:Igmi_gauss}. 
Indeed, a typical solution for $\eta$ and 
$\varepsilon$ is obtained after some tens of iterations of an
iterative substitution method.  

For the high-SNR case where $\gammabar\to\infty$ for a fixed covariance matrix
$\Rmat_{\wvec}$, the result in Example~\ref{example:mismatched_gaussian}
can be further simplified as 
shown in Example~\ref{example:mismatched_gaussian_highsnr} below.

\begin{example}
\label{example:mismatched_gaussian_highsnr}
	Let us consider the case of Gaussian signaling as given in  	Example~\ref{example:mismatched_gaussian}
	in the limit $\gammabar\to\infty$.  We assume for simplicity 
	(see, e.g., \cite{Bjornson-Zetterberg-Bengtsson-Ottersten-2013Jan})
	that $\Rmat_{\wvec} = r_{\wvec}\Imat_{N}$ and
	$r_{\vvec} = \gammabar\kappa^{2}$ where $\kappa>0$ and $r_{\wvec}>0$ are
	fixed and finite parameters.
	At high-SNR, there are two possibilities for the parameter $\sP = s / \sigmaP$ 
	in the GMI: 1) the optimal value of $\sP$ is a strictly positive constant; 
	and 2) the value of $\sP$ goes to zero when $\gammabar\to \infty$.
	For the first case, $M^{-1}I^{(s)}_{\mathsf{GMI}}(\yvec;\, \xvec) \to -\infty$ 
	so in order to obtain a consistent solution for the 
	fixed point equations, the parameter $\sP$ has to be inversely proportional
	to $\gammabar$, i.e., $\sP  = \sgamma / \gammabar$
	where $\sgamma$ is a strictly positive finite constant.  
	Then $\xi \to 0$ as $\gammabar\to \infty$,
	and the normalized GMI reduces to 
	\begin{IEEEeqnarray}{rCl}
	\label{eq:GMI_gaussian_mismatched_highsnr}
	\frac{1}{M}I^{\infty}_{\mathsf{GMI}}(\yvec;\, \xvec)
	&=&	\sup_{\sgamma>0} \bigg\{
	\frac{1}{\alpha}
	\ln \bigg(\frac{\sgamma}{\alpha \xigamma}\bigg)
	+\ln (1+\xigamma)
	\IEEEeqnarraynumspace\IEEEnonumber\\
	 &&\qquad\qquad
		+\frac{\kappa^{2} \xigamma }{1+\xigamma}
	-\frac{\sgamma\kappa^2}{\alpha}
	\bigg\}, \IEEEeqnarraynumspace
	\end{IEEEeqnarray}		
in the limit $\gammabar\to\infty$. The auxiliary parameter 
$\xigamma \triangleq \xi\gammabar > 0$ is given by 
	\begin{equation}
		\label{eq:gammaxi_gaussian_infty}
	\xigamma =
	\frac{\sgamma (1- \alpha) - \alpha  + 
	\sqrt{4 \alpha \sgamma + [ \sgamma (1- \alpha) - \alpha ]^{2}}}
	{2\alpha}.
	\end{equation}
	Compared to the finite-SNR case in Example~\ref{example:mismatched_gaussian}, 
	the GMI is now
	directly given by \eqref{eq:GMI_gaussian_mismatched_highsnr}.
\end{example}

The next example provides explicit formulas for the computation of 
GMI given finite discrete constellations, such as, PSK or QAM.

\begin{example}
\label{example:mismatched_disrete}

Let $\mathcal{A}$ be a discrete modulation alphabet 
with fixed and finite cardinality $|\mathcal{A}|$
and consider the GMI \eqref{eq:GMI_s_dependent_freeE_simple}.
Let the channel inputs $x_{m}$ be drawn independently and 
uniformly from $\mathcal{A}$.
The parameters of the decoupled channel model
in Proposition~\ref{prop:mismatched_general}
can be obtained by first solving 
$\xi$ and $\tilde{\varepsilon}$ from
\begin{IEEEeqnarray}{rCl}
\xi &=& 
\frac{\sP}{\alpha (1+\sP\tilde{\varepsilon})},\\
\tilde{\varepsilon} 
&=&
\gammabar- \int q(z) |\langle \tilde{x} \rangle_{q}|^{2} \dx z, 
\IEEEeqnarraynumspace
 \label{eq:mismatched_discrete_epsilon_postulated}
\end{IEEEeqnarray}
using the following definitions for the
decoupled estimator and the postulated channel probability
\begin{IEEEeqnarray}{rCl}
\langle \tilde{x} \rangle_{q} &=&  
\frac{1}{q(z) |\mathcal{A}|}
\sum_{\tilde{x}\in\mathcal{A}}
\tilde{x}
g(z \mid \tilde{x};\, \xi^{-1}),
\IEEEeqnarraynumspace\\
q(z) &=& \frac{1}{|\mathcal{A}|} \sum_{x\in\mathcal{A}}g(z \mid x;\, \xi^{-1}),
\end{IEEEeqnarray}
respectively.  Note that this implies solving two parameters 
from two nonlinear equations and can be done, for example,
by using an iterative substitution method.
After obtaining the solutions for 
$\xi$ (and $\tilde{\varepsilon}$), the rest of the parameters can 
be obtained by solving the two coupled equations
\begin{IEEEeqnarray}{rCl}
\label{eq:eta_simple}
\eta &=&
\frac{1}{\alpha[N^{-1}\tr(\Rmat_{\wvec}) + \varepsilon]},\\
\label{eq:mismatched_discrete_epsilon}
\varepsilon &=& 
\E \big\{|\vsym + x - \langle \tilde{x} \rangle_{q}|^{2}\big\},
\IEEEeqnarraynumspace
\end{IEEEeqnarray}
for $\eta$ and $\varepsilon$,
where the expectation is w.r.t.\ the true 
joint probability of $\{x,\vsym,z\}$.
Finally, the free energy reads
\begin{IEEEeqnarray}{l}
f_{\mathsf{RS}} (\sP) = 
\frac{1}{\alpha}\bigg(
\frac{\xi}{\eta} 
+\ln \sP+
\ln \frac{1}{\alpha \xi}
\bigg) 
-\xi\varepsilon
 + \frac{\xi(\xi-\eta)}{\eta} \tilde{\varepsilon}
\IEEEeqnarraynumspace\IEEEnonumber\\
\qquad\quad-\bigg(
\frac{\xi}{\eta}
+ \ln \frac{\pi}{\xi}
+ \int 
p(z) \ln q(z) \dx z\bigg), \;
\IEEEeqnarraynumspace
\label{eq:FreeE_mismatched_scaled_identities}
\end{IEEEeqnarray}
where we denoted
\begin{IEEEeqnarray}{rCl}
\label{eq:pz_true}
p(z) &=& \frac{1}{|\mathcal{A}|} \sum_{x\in\mathcal{A}}
  g(z \mid x;\,\eta^{-1}+ r_{\vvec}), \IEEEeqnarraynumspace
\end{IEEEeqnarray}
for the decoupled PDF of the received signal.
\end{example}

Notice that the form of $\eta$ in Example~\ref{example:mismatched_disrete}
is the same as in Example~\ref{example:mismatched_gaussian}, but the parameter 
$\varepsilon$ has now a different structure.
Compared to the Gaussian case, 
the equivalent result for IID discrete 
channel inputs looks in general more 
cumbersome.  First of all, we need to solve now two sets of 
equations instead of just one.  
They both contain terms that involve $|\mathcal{A}|$ summations
and there are also two expectations left to
evaluate,  one in 
\eqref{eq:mismatched_discrete_epsilon_postulated}
and another in \eqref{eq:mismatched_discrete_epsilon}.  However, 
both expectations involve only scalar variables.  
This is in stark contrast to
the original problem that involved computing 
$|\mathcal{A}|^{M}$ summations for every
channel and noise / distortion realization and taking 
expectation over the channel and noise that are multidimensional 
integrals.  This makes direct Monte Carlo computation of the 
GMI for discrete signaling in practice infeasible 
for large constellations and numbers of antennas.

\begin{figure*}
	\normalsize
	\setcounter{MYtempeqncnt}{\value{equation}}
	\setcounter{equation}{43}
\begin{equation}
I(\yvec;\,\xvec)=  M \ln |\mathcal{A}| - N
	- \frac{1}{|\mathcal{A}|}	\sum_{\xvec\in\mathcal{A}^{M}}
	\E_{\vvec, \wvec, \Hmat} \Bigg\{
	\ln \Bigg(
	\sum_{\xvecP\in\mathcal{A}^{M}}
	\e^{
	-[
	\Hmat (\xvec - \xvecP + \vvec) + \wvec
	]^{\herm}
	(\Rmat_{\wvec}+\Hmat \Rmat_{\vvec} \Hmat)^{-1}
	[
	\Hmat (\xvec - \xvecP + \vvec) + \wvec
	]}
	\Bigg)\Bigg\}
	\label{eq:MI_matched_discrete}
\end{equation}
\setcounter{equation}{\value{MYtempeqncnt}}
\hrulefill
\end{figure*}

\section{Matched Joint Decoding}
\label{sec:matched_decoding}

\subsection{Definition and the Special Case of Gaussian Signaling}

Let us now consider the case of matched decoding where 
the correct channel transition probability
\eqref{eq:true_decoding_pdf_expanded} is utilized at the receiver. 
The first entropy term in \eqref{eq:MI} reads
\begin{equation}
	\label{eq:matched_first_entropy}
	\E \{\ln p(\yvec \mid \xvec, \Hmat) \}
	 = -\E_{\Hmat} \{\ln \det (\Rmat_{\wvec} + \Hmat\Rmat_{\vvec}\Hmat^{\herm})\} - c,
\end{equation}
where $c = N \ln (\e \pi)$.  It should be remarked that there is still 
an expectation left w.r.t.\ the channel realizations $\Hmat$ in 
\eqref{eq:matched_first_entropy}.
This could be evaluated, for example, using Monte Carlo methods or 
 random matrix theory \cite{Tulino-Verdu-2004FnT,Couillet-Debbah-2011}.
For the special case of Gaussian inputs,
the identities in Appendix~\ref{sec:preliminaries}
allow us to partially calculate also the latter entropy term in \eqref{eq:MI},
providing the following result that is useful for Monte 
Carlo simulations.
\begin{example}
\label{example:Gaussian_simulation_matched}
Let $p(\xvec) = g(\xvec \mid \vm{0};\, \Gammamat)$.
Then, 
\begin{IEEEeqnarray}{rCl}
\frac{1}{M} I(\yvec;\,\xvec)
&=& \frac{1}{M}
\E_{\Hmat} \{\ln \det (\Rmat_{\wvec} + \Hmat(\Gammamat+\Rmat_{\vvec})\Hmat^{\herm})\}
 \IEEEeqnarraynumspace\IEEEnonumber\\
 && \quad - \frac{1}{M} 
 \E_{\Hmat} \{\ln \det (\Rmat_{\wvec} + \Hmat\Rmat_{\vvec}\Hmat^{\herm})\},
\end{IEEEeqnarray}
is the normalized ergodic MI for matched decoding. 
\end{example}

The above expression is relatively easy to compute also by 
brute-force Monte Carlo methods since there is 
only an expectation over the fading.  
Unfortunately, to the best of our knowledge, the latter entropy term 
in \eqref{eq:MI} is mathematically intractable for
rigorous methods like random matrix theory when $p(\xvec)$ is an 
arbitrary distribution 
that satisfies \eqref{eq:p_of_x}.  For example, 
given discrete inputs as in Example~\ref{example:mismatched_disrete},
calculating $\E \{ \ln p(\yvec \mid \Hmat) \}$ and combining it with
\eqref{eq:matched_first_entropy} reduces 
the MI to 
\eqref{eq:MI_matched_discrete} given at the 
top of this page.
This form is computationally very complex and 
can be evaluated using Monte Carlo methods only for small 
number of antennas and simple constellations.
To obtain a result for general input
distribution $p(\xvec)$ that has lower computational complexity,
we resort to the replica method (see Appendix~\ref{sec:replica_overview}).
As before, the results that follow have been written in a 
simplified form where the assumption of LSL is suppressed
for notational simplicity.  

\subsection{Analytical Results via the Replica Method}
\label{sec:matched_results}

\begin{prop}
\label{prop:matched_general}

\setcounter{equation}{44}

Let us write for notational convenience
\begin{IEEEeqnarray}{rCl}
	\label{eq:matched_chi}
\chi_{m} &=& x_{m} +  \vsym_{m}, \qquad m=1,\ldots,M, \IEEEeqnarraynumspace
\end{IEEEeqnarray}
where $x_{m} \sim p(x_{m})$ and 
$\vsym_{m} \sim g (\vsym_{m} \mid 0 ;\,r^{(m)}_{\vvec})$
are independent for all $m$.
Let 
\begin{IEEEeqnarray}{rCl}
	\label{eq:matched_p_z_chi}
p(z_{m} \mid \chi_{m}) &=& g (z_{m} \mid \chi_{m};\,\eta^{-1}), 
\end{IEEEeqnarray}
be a conditional PDF of an AWGN channel whose input is
\eqref{eq:matched_chi} and noise variance is $\eta^{-1}$.
The conditional mean estimator of $\chi_{m}$ received over 
this channel reads
\begin{equation}
\langle \chi_{m} \rangle = 
\frac{\E_{\chi_{m}} \{ \chi_{m}
p(z_{m} \mid \chi_{m})\}}
{\E_{\chi_{m}} \{p(z_{m} \mid \chi_{m})\}},
\end{equation}
where the parameter $\eta$ is given, along with another
parameter $\varepsilon$, as the solution to the 
coupled fixed point equations
\begin{IEEEeqnarray}{rCl}
	\label{eq:eta_matched}
\eta &=& \frac{1}{\alpha N}
 \tr \big[(\Rmat_{\wvec} +  \varepsilon \Imat_{N})^{-1}\big], \IEEEeqnarraynumspace\\
 \label{eq:matched_mmse}
 \varepsilon &=&
 \frac{1}{M}\sum_{m=1}^{M}
 \big[\gammabar_{m}+r^{(m)}_{\vvec} -
 \E |\langle \chi_{m}\rangle|^{2}\big].
 \IEEEeqnarraynumspace
\end{IEEEeqnarray}
If we also define a second set of 
parameters $\eta'$ and $\varepsilon'$ that 
are solutions to the coupled fixed point equations
\begin{IEEEeqnarray}{rCl}
	\label{eq:matched_etaP}
\eta' &=& 
\frac{1}{\alpha N}
\tr \big[(\Rmat_{\wvec} +  \varepsilon' \Imat_{N})^{-1}\big], \IEEEeqnarraynumspace\\
\varepsilon' 
&=& 
\frac{1}{M}
\sum_{m=1}^{M}\frac{ r^{(m)}_{\vvec}}{1+\eta'r^{(m)}_{\vvec}},
\IEEEeqnarraynumspace
	\label{eq:matched_epsilonP}
\end{IEEEeqnarray}
the per-stream MI is finally given by
\begin{IEEEeqnarray}{l}
 \frac{1}{M}I(\yvec;\, \xvec)=
\frac{
\ln \det (\Rmat_{\wvec} +  \varepsilon \Imat_{N})
\!-\! \ln \det (\Rmat_{\wvec} +  \varepsilon' \Imat_{N})
}{\alpha N}
  \IEEEnonumber\\
 - (\eta \varepsilon - \eta' \varepsilon')  
  + \frac{1}{M} \sum_{m=1}^{M} 
\big[I(z_{m};\,\chi_{m}) - \ln (1+\eta' r^{(m)}_{\vvec})\big],
\IEEEeqnarraynumspace 
\label{eq:matched_mi}
\end{IEEEeqnarray}
where 
\begin{equation}
	\label{eq:Izchi}
I(z_{m};\,\chi_{m})
= - 1 - \ln \frac{\pi}{\eta} - \int p(z_{m}) \ln p(z_{m}) \dx z_{m},
\end{equation}
is the MI of the Gaussian channel defined by 
\eqref{eq:matched_chi}
and \eqref{eq:matched_p_z_chi}. 
\end{prop}

\begin{proof}
The result can be obtained using Appendix~\ref{sec:replica_overview}
for two separate MIMO channels.  For the first one, 
we replace everywhere 
$\xvec_{a} \to \xvec_{a} + \vvec_{a}, a = 0,1,\ldots,u$ and
and an application of the RM provides the equations 
\eqref{eq:matched_chi}--\eqref{eq:matched_mmse}. 
The formulas \eqref{eq:matched_etaP}--\eqref{eq:Izchi}, on the other hand, 
are obtained by substituting 
$\xvec_{a} \to \vvec_{a}, a = 0,1,\ldots,u$ in Appendix~\ref{sec:replica_overview}.
\end{proof}

Just like Proposition~\ref{prop:mismatched_general} in 
Section~\ref{sec:GMI}, Proposition~\ref{prop:matched_general} 
is valid for any input distribution that satisfies \eqref{eq:p_of_x}.
The solutions to the coupled equations 
\eqref{eq:eta_matched}~and~\eqref{eq:matched_mmse}
as well as \eqref{eq:matched_etaP}~and~\eqref{eq:matched_epsilonP}
can be obtained numerically, e.g., using an iterative substitution method.

For concreteness, we again give examples 
for Gaussian and discrete signaling when the
noise plus distortion is spatially white 
$\Rmat_{\vvec} = r_{\vvec}\Imat_{M}$ and 
transmit power is uniformly allocated
$\Gammamat = \gammabar\Imat_{M}$. This makes
the channels $m=1,2,\ldots,M$ identically 
distributed so we omit the subscript
$m$ in the following.

\begin{example}
\label{example:matched_gaussian}

Let $\Rmat_{\vvec} = r_{\vvec}\Imat_{M}$ and consider
the special case of Gaussian inputs
$p(\xvec) 
= g(\xvec \mid \vm{0}; \, \gammabar\Imat_{M})$. 
Then 
\begin{IEEEeqnarray}{rCl}
I(z;\,\chi)
&=& \ln \big[1 +  \eta(\gammabar + r_{\vvec})\big],
\IEEEeqnarraynumspace \\
\varepsilon &=& 
\frac{\gammabar + r_{\vvec}}{1+
\eta(\gammabar + r_{\vvec})},
\IEEEeqnarraynumspace 
\end{IEEEeqnarray}
and the rest of the parameters are given in 
Proposition~\ref{prop:matched_general}.
\end{example}

We next consider the high-SNR case 
$\gammabar\to\infty$ as in Example~\ref{example:mismatched_gaussian_highsnr}
and compare it to the result obtained
in \cite{Bjornson-Zetterberg-Bengtsson-Ottersten-2013Jan} 
using completely different mathematical methods.

\begin{example}
\label{example:matched_gaussian_highsnr}
For the case 
$\Rmat_{\wvec} = r_{\wvec}\Imat$,
$\Rmat_{\vvec} = \kappa^{2} \gammabar \Imat$
(see, e.g., \cite{Bjornson-Zetterberg-Bengtsson-Ottersten-2013Jan})
we find that 
if $\alpha \leq 1$
then $\gammabar\to \infty$ yields
$\eta = \eta'$ 
and 
$\varepsilon = \varepsilon'$. 
The high SNR limit is therefore
\begin{equation}
\frac{1}{M}I^{\infty}(\yvec;\, \xvec)=
\log\bigg(\frac{1+\kappa^{2}}{\kappa^{2}}\bigg), 
\qquad \alpha \leq 1.
\label{eq:emil_limit_low}
\end{equation}
For the case $\alpha > 1$, both $\eta$ and $\eta'$ tend to zero 
at high SNR while $\varepsilon$ and $\varepsilon'$ grow without bound.
This is not yet sufficient to solve \eqref{eq:matched_mi}.  
However, combining this with the relations
$\eta'\varepsilon'=\eta\varepsilon$ and
$\varepsilon' = \varepsilon \frac{\kappa^{2}}{1+\kappa^{2}}$,
that hold in the limit $\gammabar\to\infty$ for 
$\alpha>1$, provides 
the second part of the high SNR result
\begin{equation}
\frac{1}{M}I^{\infty}(\yvec;\, \xvec)=
\frac{1}{\alpha}\log\bigg(\frac{1+\kappa^{2}}{\kappa^{2}}\bigg), 
\qquad \alpha > 1.
\label{eq:emil_limit_high}
\end{equation}
The asymptotic mutual information expressions 
in \eqref{eq:emil_limit_low} and \eqref{eq:emil_limit_high} coincide
exactly with the results obtained previously in
\cite{Bjornson-Zetterberg-Bengtsson-Ottersten-2013Jan},
as expected.
\end{example}

\begin{example}
\label{example:matched_disrete}

If the channel inputs are from a discrete alphabet
$\mathcal{A}$ as in Example~\ref{example:mismatched_disrete},
the parameter $\varepsilon$ in \eqref{eq:matched_mmse} is obtained 
using
 \begin{IEEEeqnarray}{l}
 \langle \chi \rangle 
=
 \frac{1}{p(z)}
\sum_{x\in\mathcal{A}} \bigg[\frac{1}{|\mathcal{A}|}
g(z \mid x;\,\eta^{-1}+ r_{\vvec})
 \bigg(\frac{x
 + \eta r_{\vvec} z}{1+\eta r_{\vvec}}\bigg)
 \bigg], \IEEEeqnarraynumspace \\
 \E|\langle \chi\rangle|^{2} 
= \int p(z) \E\big\{|\langle \chi\rangle|^{2}\big\}
 \dx z, 
 \end{IEEEeqnarray}
 in Proposition~\ref{prop:matched_general}.
 Here $p(z)$ is given by \eqref{eq:pz_true} and 
 $\langle \chi \rangle$ denotes the conditional mean 
 estimator of \eqref{eq:matched_chi} from the observations
\eqref{eq:matched_p_z_chi}. 
The related MI term reads by definition
\begin{equation}
\label{eq:matched_mi_qam}
I(z;\, \chi ) = \ln \bigg(\frac{\eta}{\e \pi}\bigg)
-  \int p(z) \ln  p(z) \dx z.
\end{equation}
Both \eqref{eq:matched_mmse} and \eqref{eq:matched_mi_qam} need, in general, 
to be solved numerically.
\end{example}

\begin{figure*}[tb]
\centering
\subfloat[Gaussian signaling]{\includegraphics[width=0.97\columnwidth]{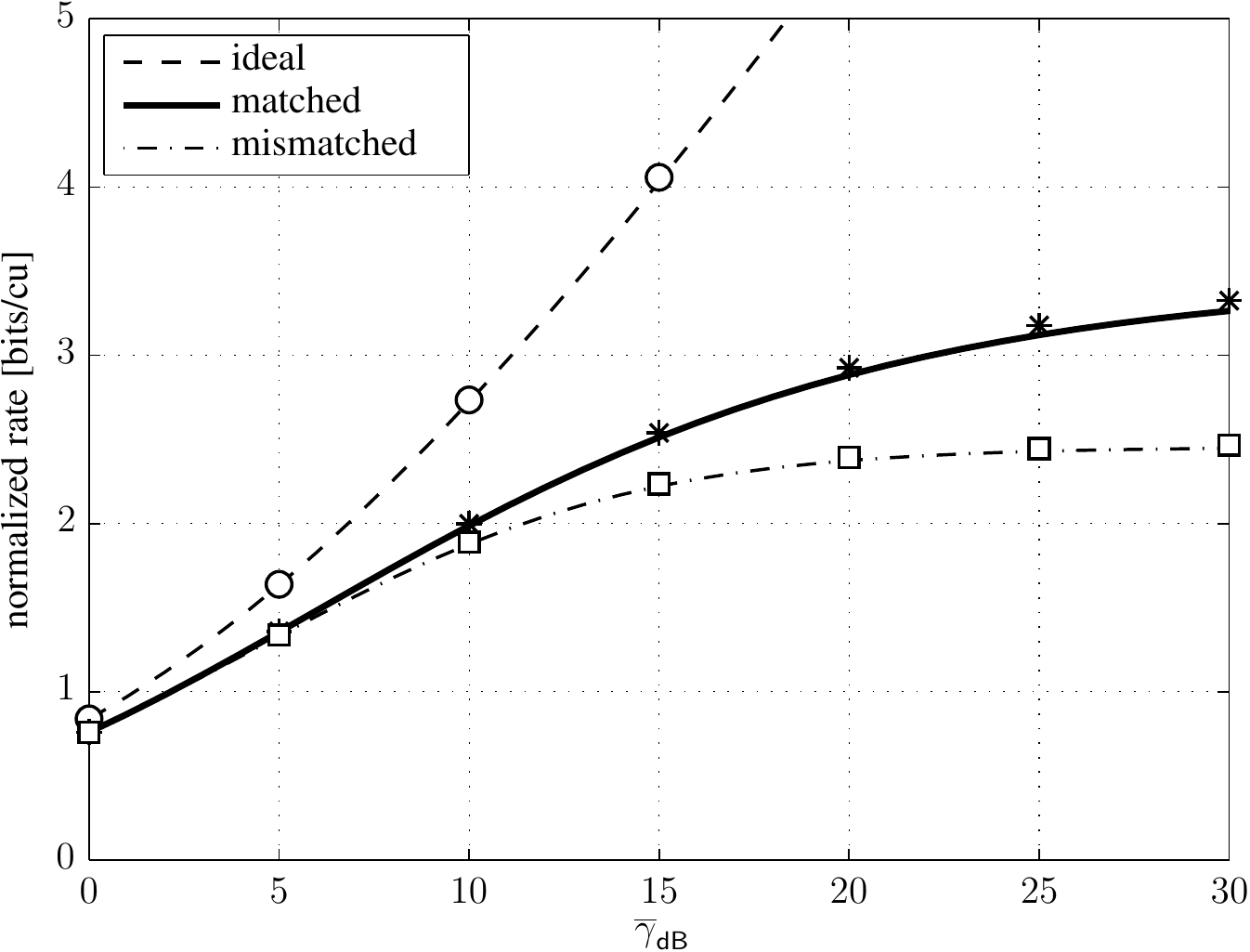}
\label{fig:gaussian_sim_1}}
\hfil
\subfloat[Discrete signaling]{\includegraphics[width=0.97\columnwidth]{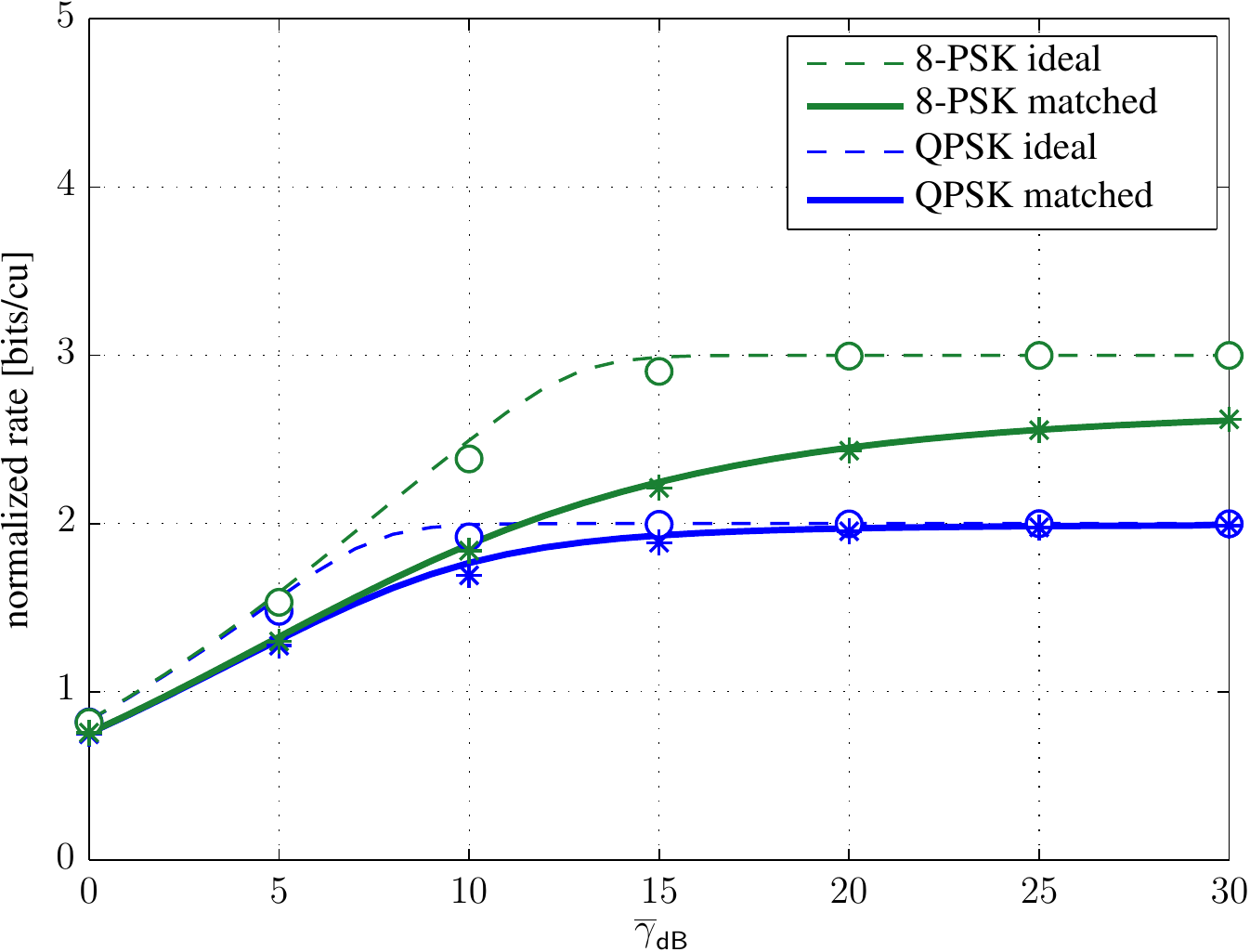}
\label{fig:discrete_sim_1}}
\caption{Normalized rate $M^{-1} I(\yvec;\, \xvec)$ in bits per
channel use (cu) vs.\ SNR for MIMO transmission. Lines for replica results and markers for Monte Carlo simulations for $M = N = 4$
antenna configuration.
Selected cases of ideal hardware $\mathsf{EVM} = -\infty$~dB and 
hardware impairments ($\mathsf{EVM} = -10$~dB) with
matched and mismatched decoding are plotted.}
\label{fig:discrete_Gauss_sims}
\end{figure*}

\section{Numerical Examples}

In the following, assume for simplicity that $\Gammamat = \gammabar \Imat$,
$\Rmat_{\wvec} = \Imat$ and
$\Rmat_{\vvec} = \kappa^{2} \gammabar \Imat$,
where $\kappa = 10^{\mathsf{EVM} / 20}$
and $\mathsf{EVM}$ denotes the EVM of the transmitter in decibels.
The SNR without transmit-side noise is therefore simply $\gammabar$,
or in decibels, $\gammabar_{\mathsf{dB}} = 10 \log_{10}(\gammabar)$.
Furthermore, all cases assume a symmetric antenna setup 
$\alpha = M/N = 1$ for simplicity. 

The first numerical experiment plotted in Fig.~\ref{fig:discrete_Gauss_sims} 
examines the accuracy of the asymptotic analytical results 
when applied to finite-sized systems.
The EVM is fixed to a rather pessimistic value
$\mathsf{EVM}=-10$~dB
to highlight the differences between the ideal and 
imperfect hardware configurations.  
The normalized rate is shown using the asymptotic
replica analysis (lines) and 
Monte Carlo simulations (markers) for a finite-size 
symmetric antenna setup with $M = N = 4$.
In the case of Gaussian signaling, plotted  in 
Fig.~\ref{fig:gaussian_sim_1}, the analytical
approximations for the normalized rate
$M^{-1} I(\yvec;\, \xvec)$ given by 
Examples~\ref{example:mismatched_gaussian}~and~\ref{example:matched_gaussian}
 are quite good when compared to the finite size simulations based on
Examples~\ref{example:Gaussian_simulation_mismatched}~and~%
 \ref{example:Gaussian_simulation_matched}.
For discrete signaling depicted in Fig.~\ref{fig:discrete_sim_1}
we have plotted only the case of matched decoding due to the computational 
complexity of Monte Carlo simulations in the mismatched case.
The gap between asymptotic result presented in Example~\ref{example:matched_disrete} 
and Monte Carlo averaging of \eqref{eq:MI_matched_discrete}
is similar to the Gaussian case for both constellations.   
Figure~\ref{fig:discrete_Gauss_sims} shows that 
the analytical approximation given by the replica method 
is reasonably good already at $M = N = 4$, even though 
formally the limit $M,N \to \infty$ is required by
the analysis.  
Note that Monte Carlo simulation of
\eqref{eq:MI_matched_discrete} has exponential computational complexity
and the system size cannot be increased much higher than $M=4$.
Therefore, the rest of the examples are generated using only the analytical
results given in the previous sections.

\begin{figure*}[tb]
\centering
\subfloat[Normalized rate $M^{-1} I(\yvec;\, \xvec)$ given ideal hardware (dashed lines) or non-ideal hardware and  matched decoding (solid lines).]%
{\includegraphics[width=0.97\columnwidth]{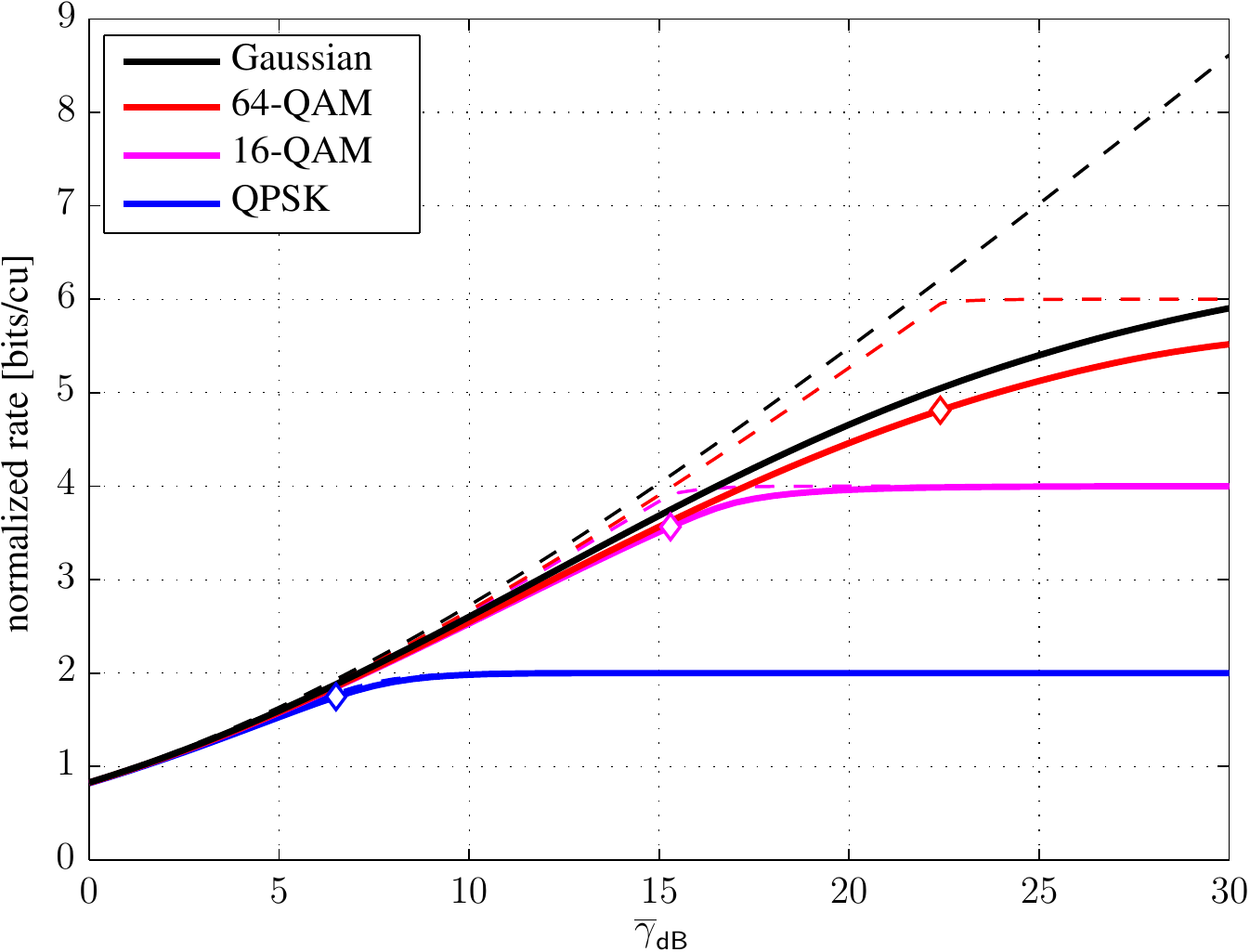}
\label{fig:matched_comp_1}}
\hfil
\subfloat[Rate loss percentage compared to ideal hardware for matched (solid lines) and mismatched (dash-dotted lines) decoding. ]{\includegraphics[width=0.97\columnwidth]{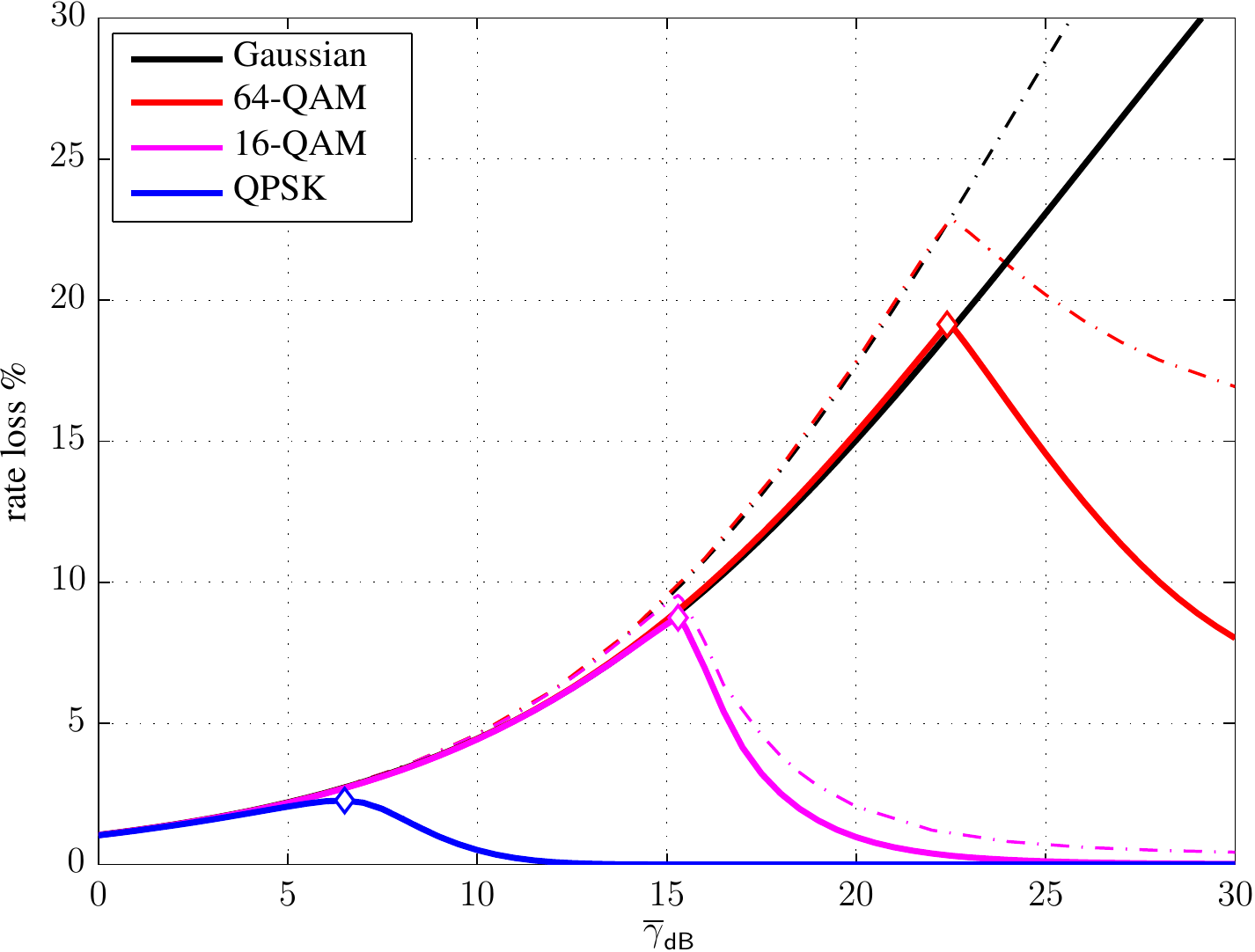}
\label{fig:rateloss_comp_1}}
\caption{
Performance of a MIMO system with $M = N$ antennas and 
given ideal 
($\mathsf{EVM} = -\infty$~dB) or non-ideal hardware 
($\mathsf{EVM} = -20$~dB) for different signaling 
methods.  Markers depict the points where discrete constellations and matched decoding with hardware impairments experience the maximum rate losses compared to the ideal cases.}
\label{fig:rateplot_replicas}
\end{figure*}

Figure~\ref{fig:rateplot_replicas} 
illustrates the performance of an $M = N$ MIMO system
for a more realistic EVM value $\mathsf{EVM}=-20$~dB.
For the case of matched decoding 
we used 
Examples~\ref{example:matched_gaussian}~and~\ref{example:matched_disrete},
while 
Examples~\ref{example:mismatched_gaussian}~and~\ref{example:mismatched_disrete}
were used to obtain the curves representing mismatched decoding.
In Fig.~\ref{fig:matched_comp_1}, the normalized rate 
$M^{-1} I(\yvec;\, \xvec)$ is depicted as a function of SNR
$\gammabar$ in decibels.  For clarity of presentation, 
we have plotted only the ideal case and the case of 
non-ideal hardware with matched decoding.  The Gaussian curves 
(black lines) here are the same as the simulation curves in  \cite[Fig.~2]{Bjornson-Zetterberg-Bengtsson-Ottersten-2013Jan}
given the parameter value $\kappa = 0.1$.  Apart from 
64-QAM and Gaussian signaling, the figure seems to imply 
that lower order constellations exhaust the 
source entropy before the transmit-side noise has any significant
effect for this choice of EVM.
To see more clearly the effect of transmit noise, 
Fig.~\ref{fig:rateloss_comp_1} shows the rate loss 
(in percentage) for the case with transmit noise
$\mathsf{EVM}=-20$~dB 
when compared to the ideal case $\mathsf{EVM}=-\infty$~dB.
The solid lines represent again matched decoding while 
dash-dotted lines are for mismatched decoding.  As expected, 
mismatched decoding reduces the achievable rate
 when compared to matched decoding, but the effect is 
 relatively minor when compared to the total rate loss
caused by the presence of transmit noise itself.  
The markers depict the points where maximum relative
rate loss is experienced for matched decoding.
The same markers 
are also plotted in Fig.~\ref{fig:matched_comp_1} for 
comparison.

In Fig.~\ref{fig:EVMgauss} we have plotted the asymptotic high-SNR 
results given in Examples~\ref{example:mismatched_gaussian_highsnr}%
~and~\ref{example:matched_gaussian_highsnr}.  Note that given a finite value of 
$\mathsf{EVM}$, the normalized rates 
for matched and mismatched decoding have a gap in this case.
For more realistic, but still quite high SNR values of $20$~dB and $30$~dB, the two 
decoding strategies converge to the same value roughly when 
$\gammabar_{\mathsf{dB}} < -\mathsf{EVM}$.
The apparent discrepancy is explained by recalling that the 
asymptotic cases assume 
$\gammabar\to\infty$ for a fixed and nonzero EVM and, thus, as a finite
SNR approximation implies
$\gammabar \gg 1/\kappa^{2}$.  As may be observed
from the lower right corner of the figure, 
the SNR values $20$~dB and $30$~dB have also a 
similar behavior near $\gammabar \gg 1/\kappa^{2}$.
Thus, the high-SNR result is consistent with the finite-SNR
cases.

It is important to guarantee certain performance when designing a system.
The maximum EVM that leads to at most 5\% rate loss 
(as compared to having ideal hardware)
for a fixed input distribution and different given SNRs
is plotted in Fig.~\ref{fig:maxEVM}.  
For Gaussian signaling we have plotted 
both the matched and mismatched cases while discrete cases 
assume matched joint decoding for simplicity.  
As expected, the EVM requirement for Gaussian signaling is a monotonically
decreasing, but not linear, function of SNR.  A simple linear approximation 
that provides a lower bound for the case of 
Gaussian signaling with matched decoding is given by
\begin{equation}
	\mathsf{EVM} = -0.7\cdot\gammabar_{\mathsf{dB}}-13,
\end{equation}
in decibels for the depicted region.  
This can be used as a simple rule-of-thumb for worst-case 
maximum allowed EVM in the system, although we recommend that EVM target values obtained 
in this way are always rounded down to  $1$--$5$~dB precision to include extra safety margin.
For discrete constellations, 
the EVM requirement first follows the Gaussian case but then starts to get looser 
at higher SNRs.  This is expected, as can be observed from Fig.~\ref{fig:matched_comp_1},
since the maximum achievable rate for a discrete constellation saturates at a certain SNR
when the input distribution runs out of entropy. 
After this point, the rate loss can be held fixed for increasing SNR by increasing 
the transmit-side noise variance, or EVM, accordingly.

\begin{figure}[tb]
\centering
\includegraphics[width=\columnwidth]{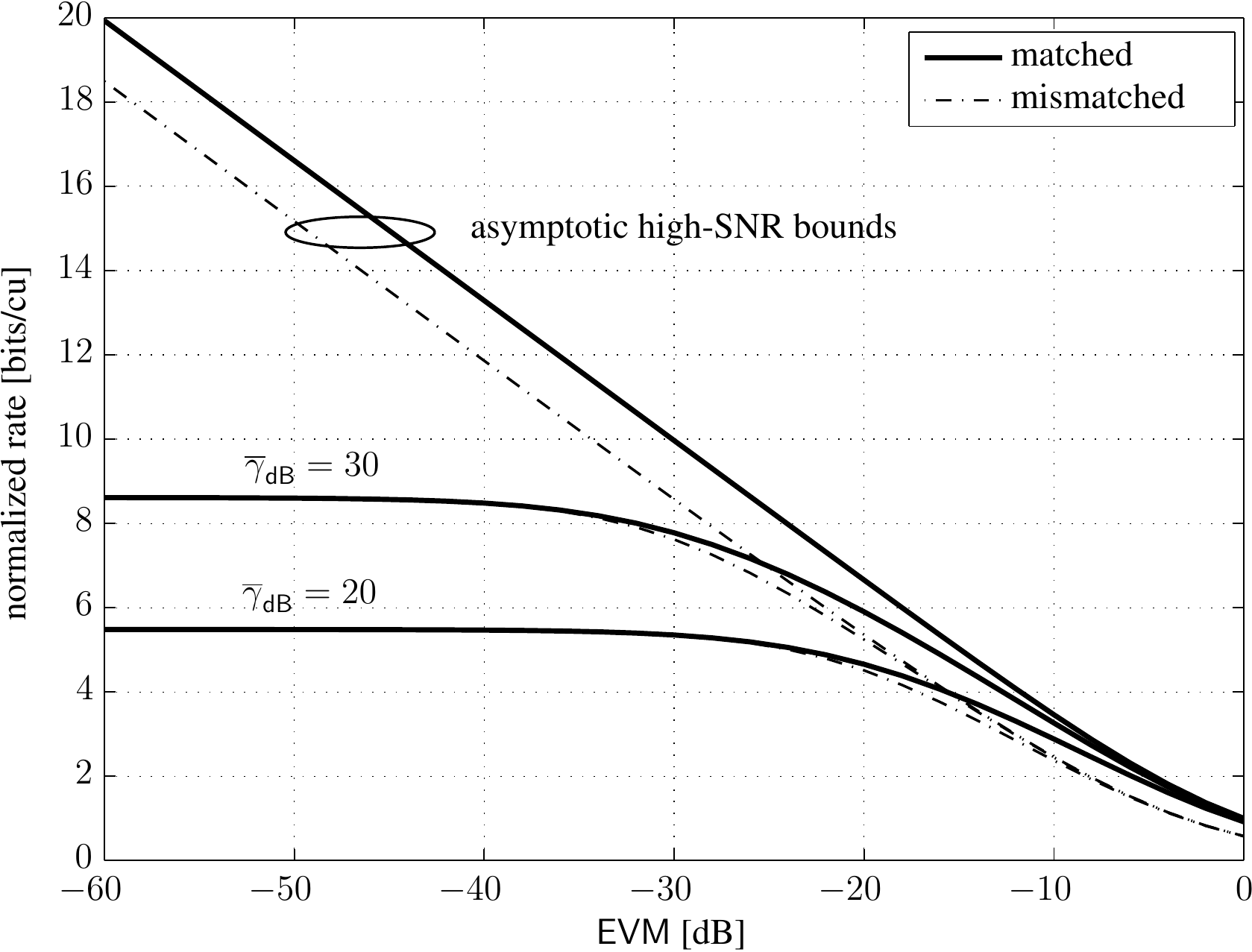}
\caption{
Normalized rate $M^{-1} I(\yvec;\, \xvec)$ in bits per
channel use vs.\ EVM in decibels 
for MIMO transmission with Gaussian signaling. 
Solid lines for matched decoding and dash-dotted lines for 
mismatched decoding.
\label{fig:EVMgauss}}
\end{figure}
\begin{figure}[tb]
\centering
\includegraphics[width=\columnwidth]{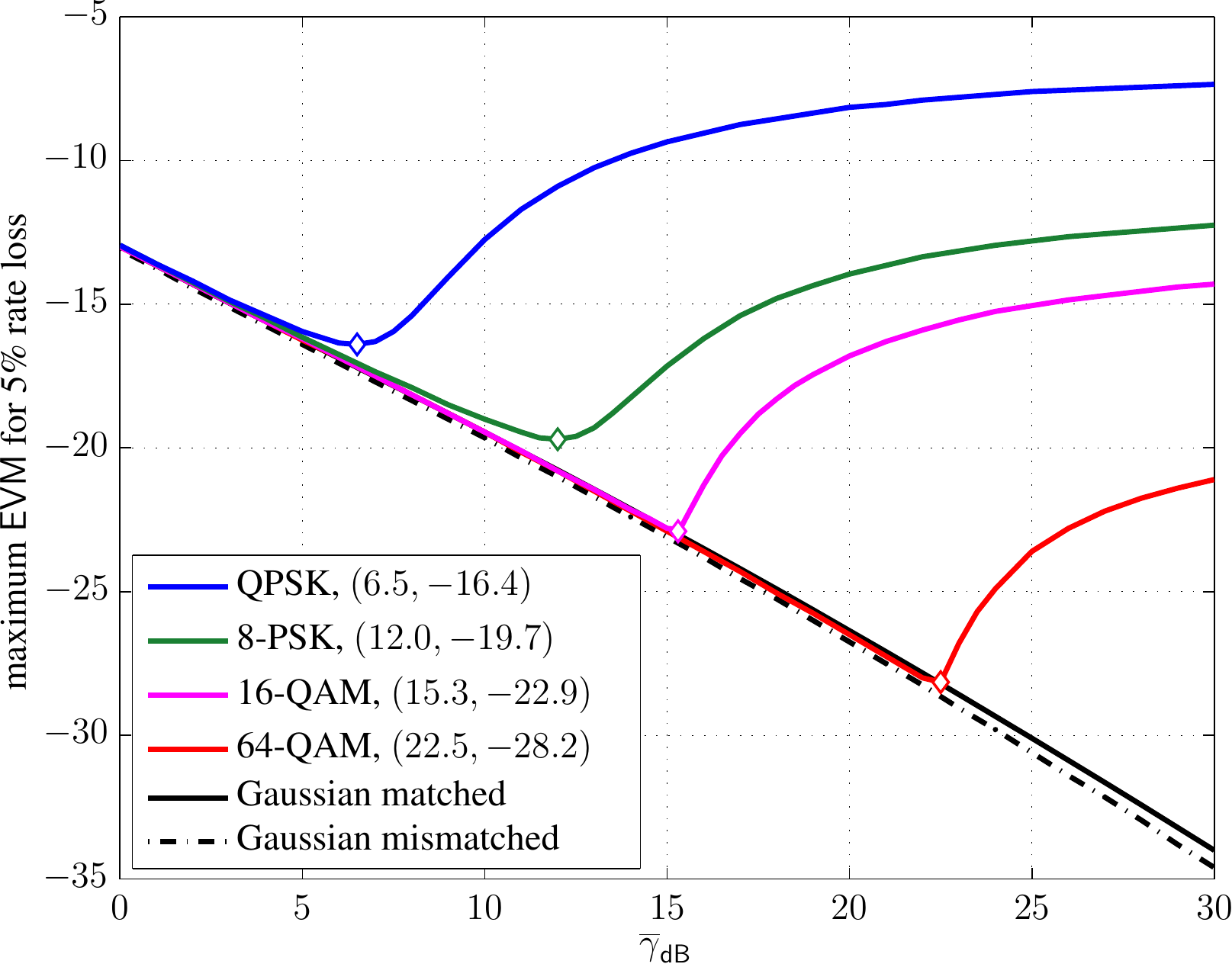}
\caption{Maximum allowed EVM in decibels for matched decoding so that 
the system experiences at most 5\% loss in rate compared to 
the case with ideal hardware ($\mathsf{EVM} = -\infty$~dB).  Markers 
depict the worst case EVM requirement for the discrete constellations
and parenthesis in the legend provide the respective values as 
$(\gammabar_{\mathsf{dB}},\mathsf{EVM})$.
All discrete cases correspond to matched joint decoding at
the receiver.}
\label{fig:maxEVM}
\end{figure}

\section{Conclusions and Future Work}

Considering a `binoisy' channel model, we have derived asymptotic expressions 
for the achievable rate of SU-MIMO systems suffering from transceiver 
hardware impairments. For matched decoding, where the 
receiver is designed and implemented explicitly based on the generalized 
system model, expressions for the ergodic 
mutual information between the channel inputs and outputs have been given.
In addition, a simplified receiver that neglected the hardware imperfections
and performed mismatched detection and decoding has been studied 
via generalized mutual information. 
The mathematical expressions provided in the paper cover practical discrete modulation
schemes, such as, quadrature amplitude modulation, as well as Gaussian signaling.
The numerical results showed that for realistic system parameters, the effects of
transmit-side noise and mismatched decoding become significant only at high modulation
orders.  Furthermore, the effect of mismatched decoding was found to be relatively 
minor compared to the total rate loss caused by the presence of transmit noise itself.
The results were also used to identify the maximum EVM values that allows for certain system operation.

\subsection{Future Work}

For the ease of exposition, the present paper considered the analysis of a 
relatively simple SU-MIMO system where the channel had IID Gaussian elements.
An extension of the replica analysis to Rayleigh fading channels with Kronecker 
correlation can be done by following, e.g., 
the derivations in \cite{Wen-Wong-Chen-TCOM2007}.  
Establishing the effects of transmit-side
noise for the case of correlated channel is an important avenue 
for future work.

As a further extension, it is important to investigate whether similar 
phenomena as observed in the present paper are present also for more complicated 
signal models with discrete channel inputs.  Such systems already analyzed 
in the ideal setting with the replica method include, for example, 
multiuser MIMO and base station collaboration \cite{Wen-Wong-2010}, 
channels with interference and precoding \cite{Girnyk-Vehkapera-Rasmussen-2014} 
and $K$-hop relay channels \cite{Girnyk-Vehkapera-Rasmussen-2014arxiv}.  
Combining the ideas from the present paper and 
\cite{Wen-Wong-2010, Girnyk-Vehkapera-Rasmussen-2014,Girnyk-Vehkapera-Rasmussen-2014arxiv}
would provide a possible approach to solving such cases.

\appendices

\section{Useful Results}
\label{sec:preliminaries}

Here we collect useful results that are 
used often in the paper.   All matrix operations below
are implicitly assumed to be well-defined.
The \emph{Gaussian integration} formula for 
vector $\vm{x}\in\mathbb{C}^{N}$ is given by
(see, e.g., \cite[Appendix~I]{Moustakas-Simon-2007})
\begin{equation}
\label{eq:Gint}
\frac{1}{\pi^{N}}\int \e^{-\vm{x}^{\herm} \vm{M} 
\vm{x}+ 2 \Re\{\vm{b}^{\herm} \vm{x}\}} \dx \vm{x} 
= \frac{1}{\det(\vm{M})}\e^{\vm{b}^{\herm}\vm{M}^{-1}\vm{b}},
\end{equation}
and used in
Sections~\ref{sec:system_model}~--~\ref{sec:matched_decoding} 
and Appendix~\ref{sec:replicas}.  Similarly, the 
\emph{matrix inversion lemma} \cite{Bernstein-2009}
\begin{IEEEeqnarray}{l}
(\vm{W}^{-1}+\vm{U}\vm{T}^{-1}\vm{V}^{\herm})^{-1} \IEEEnonumber\\
\qquad \qquad = \vm{W}-\vm{W}\vm{U}
(\vm{T}+\vm{V}^{\herm}\vm{W}\vm{U})^{-1}\vm{V}^{\herm}\vm{W},
\IEEEeqnarraynumspace
\label{eq:matrix_id1}
\end{IEEEeqnarray}
and the related \emph{determinant identity} 
\begin{IEEEeqnarray}{l}
	\label{eq:det_id1}
\det(\vm{W}^{-1}+\vm{U}\vm{T}^{-1}\vm{V}^{\herm}) \IEEEnonumber\\
\qquad \qquad =\det(\vm{T}+\vm{V}^{\herm}\vm{W}\vm{U})\det(\vm{W}^{-1})\det(\vm{T}^{-1}),
\IEEEeqnarraynumspace
\end{IEEEeqnarray}
are employed several times in the paper.

\section{Replica Method}
\label{sec:replica_overview}

Consider a function $\funcZ$ that maps RVs
to real numbers%
\footnote{In the following we refrain differentiating random variables and their realizations
for notational convenience.  Also, 
$\funcZ$ and, as a result, $\outF$ can depend on some parameters 
(non-random variables) that are not explicitly stated.} and
define two sets of RVs, $\varV\in\setV$ and $\varX\in\setX$, with 
joint probability $P_{\varV,\varX}$.  Assume
for convenience that $P_{\varV,\varX}$ can be described in terms of a joint 
PDF $p(\varV,\varX)$ and denote the marginal PDFs of $\varX$ and $\varV$
$p_{\varX}(\varX)$ and $p_{\varV}(\varV)$, respectively.  
Then, both in statistical mechanics and communication theory, we often encounter
a formula
\begin{IEEEeqnarray}{rCl}
	\label{eq:replica_Eln_1}
	\outF &=& - \frac{1}{M} \E_{\varV} \big\{\ln  \E_{\varX}\{\funcZ(\varV,\varX)\} \big\}
	\IEEEnonumber\\
	 &=& - \frac{1}{M} \int_{\setV} p_{\varV}(\varV) \ln \funcZ(\varV) \, \dx \varV,
	 \IEEEeqnarraynumspace
\end{IEEEeqnarray}
where $\funcZ(\varV) =  \int_{\setX} p_{\varX}(\varX)\funcZ(\varV,\varX) \dx \varX$.
	In physics jargon, the variables $\varV$ are said to be \emph{quenched}
	and the quantity \eqref{eq:replica_Eln_1} is the average \emph{free energy density}
	of a system whose \emph{partition function} is $\funcZ(\varV)$.
Two concrete examples of \eqref{eq:replica_Eln_1} are:
\begin{enumerate}
	\item Let $\funcZ(\varV,\varX) =
	g(\yvec \mid \Hmat \xvec;\, \Rmat_{\wvec})$ be the conditional 
	PDF of the observation in an ideal MIMO channel with
	$\varV = \{\yvec,\Hmat\}$ and $\varX = \{\xvec\}$, where
	$\xvec$ has IID\ elements from a discrete modulation 
	set $\mathcal{A}$, such as PSK or QAM.
	Then \eqref{eq:replica_Eln_1} represents a normalized version of the second term 
	in \eqref{eq:MI}, namely, the (normalized) 
	total \emph{entropy} of the received signal $\yvec$
	given a realization of $\Hmat$ and
	averaged over all possible realizations of $\Hmat$.	
	\item Let $\funcZ(\varV,\varX) = 
	\e^{\beta\vm{\sigma}^{\herm}\vm{J}\vm{\sigma}}$, where 
	$\beta>0$ denotes the inverse temperature, 
	$\varV = \vm{J} \in \mathbb{R}^{M \times M}$ a coupling matrix
	and $\varX = \vm{\sigma}\in\{\pm 1\}^{M}$ a 
	spin configuration.  If $p_{\varV}(\varV)$ %
	is a uniform probability over $\vm{\sigma}$
	and $\vm{J}$ has, e.g., IID\ Gaussian elements, then
	\eqref{eq:replica_Eln_1} is the average \emph{free energy density} of 
	a mean-field Ising spin glass in the absence of external field
	(up to trivial constants).
\end{enumerate}
In both cases, $\outF$ captures important properties of the 
system at hand 
and obtaining a computable formula for \eqref{eq:replica_Eln_1} would be of great 
interest.  This seems infeasible though since the number of 
terms in the expectation is exponential in $M$.

\subsection{Outline of the Replica Method}  

One method for solving \eqref{eq:replica_Eln_1} is the \emph{replica method} (RM)
from equilibrium statistical mechanics.  While the RM is extremely versatile,
it unfortunately lacks mathematical rigor 
in some parts (see, e.g., \cite{Tanaka-2002Nov,Nishimori-2001, Mezard-Montanari-2009}).
However, due to its success both in physics and engineering, it is generally 
agreed to be at least a valuable starting point for analysis of problems that 
seem otherwise too difficult to handle.  
A cursory overview of 
literature about the RM inside a specific field or topic may paint the picture 
that the RM is a fixed set of mathematical methods which can be applied 
to any suitable problem at hand.  This is not entirely accurate and
conceptually the RM can be seen more like a systematic way 
of turning a very difficult problem into a more manageable one than a set of specific
tools that actually solve the problem.  Indeed, the mathematical 
methods that are used at different stages of the RM can often be chosen from 
a variety of choices, although it is very common to have 
some form of \emph{large deviations theory} as part of the 
analysis (see \textbf{Step 2} below).    
Thus, instead of trying to be entirely general, we describe next (one form of) 
the steps taken in the RM in the context of the first example above. 

\begin{step}[Replica trick]
	Consider  \eqref{eq:replica_Eln_1} and write it as
	\begin{IEEEeqnarray}{rCl}
	\funcF &=&
	-\frac{1}{M}\lim_{u\to 0^{+}}
	\frac{\partial}{\partial u} 
	\ln \E_{\varV}\{[\funcZ(\varV)]^{u}\}
	\IEEEnonumber\\
	&=&
	- \frac{1}{M}\lim_{u\to 0^{+}}
	\frac{\partial}{\partial u} 
	\ln
	\E_{\varV}\bigg\{
	\bigg(\sum_{\xvec\in \mathcal{A}^{M}} p_{\varX}(\xvec) 
	\funcZ(\varV,\xvec)  \bigg)^{u} %
	\bigg\} \IEEEeqnarraynumspace\IEEEnonumber\\
	&=& -  \frac{1}{M}\lim_{u\to 0^{+}}
	\frac{\partial}{\partial u} \ln \Xi(u),
	\label{eq:F_alternative_1}	
	\end{IEEEeqnarray}	
where $u\in\mathbb{R}$ and we denoted 
$\Xi(u) = \E_{\varV}\{[\funcZ(\varV)]^{u}\}$.  
Then, \emph{assume} that we can treat $u$ as an integer when 
we take the expectation, namely,
\begin{IEEEeqnarray}{rCl}
	\label{eq:integer_u_1}
	\Xi(u) &=& 
	\E_{\varV}\bigg\{
	\prod_{a=1}^{u} 
	\sum_{\xvec_{a}\in\mathcal{A}^{N}} 
	p_{\varX}(\xvec_{a}) \funcZ(\varV,\xvec_{a}) 
	\bigg\} \\ 
	&=& 
	\frac{1}{\pi^{u N}(\det \Rmat_{\wvec})^{u}}	 \IEEEnonumber\\
	&&\times
	\E_{\varV}\bigg\{
	\sum_{\{\xvec_{a}\}} 
	\prod_{a=1}^{u}
	\bigg[
	\e^{-(\yvec-\Hmat \xvec_{a})^{\herm}
	\Rmat_{\wvec}^{-1}(\yvec-\Hmat \xvec_{a})
	}
	 p_{\varX}(\xvec_{a})
	\bigg] \bigg\}, \IEEEnonumber
 	\label{eq:integer_u_2} 
\end{IEEEeqnarray}	
where the summation in the last expression is over 
the set $\{\xvec_{a}\}_{a=1}^{u}$.  After taking the expectations, 
if we manage to write \eqref{eq:integer_u_1} in a form that 
does not explicitly force $u$ to be an integer,
\emph{invoke analytical continuity} to extend $u$ to real numbers. 
\end{step}

The step above is at the very heart of the RM.  It is important 
to realize that the equalities in 
\eqref{eq:F_alternative_1} are provably true
if differentiation under the integral sign is permitted and $u\in\mathbb{R}$.
The part lacking rigorous mathematical justification is
\eqref{eq:integer_u_1}, especially when 
combined with the next two steps.  Somewhat surprisingly, however, 
the end results of RM can sometimes be proved to be exact.  Examples
of such cases are: MIMO channel with Gaussian inputs, random energy model (REM)
and Sherrington-Kirkpatrick model of spin glasses 
(see, e.g., \cite{Tanaka-2002Nov, Guo-Verdu-2005Jun, Nishimori-2001, 
Mezard-Montanari-2009}
and references therein).

\begin{step}[Large system limit]
	Let the system approach the LSL, that is, the dimensions of the 
	channel matrix $\Hmat$ grow without bound at a finite  and fixed 
	ratio $\alpha = M / N > 0$.  Furthermore, assume that the 
	limits w.r.t.\ $u$ and $M$ commute, 
	so that we can first calculate the expectations in \eqref{eq:integer_u_1}
	in the LSL and then let $u\to 0$, as in \eqref{eq:F_alternative_1}.
\end{step}

The LSL assumption is natural in equilibrium 
statistical mechanics (e.g.\ the second example above), where the systems 
contain usually very large numbers of interacting particles $M$. 
In communication theory, the equivalent would be, e.g., a MIMO 
systems with large antenna arrays or a CDMA with
large number of simultaneous users.
It is in fact quite common to write the LSL assumption
directly as a part of the replica trick in \eqref{eq:F_alternative_1}. 
The steps are separated here since the replica trick could also be used 
for finite sized systems.  Due to mathematical difficulty of such cases, 
however, both steps are  usually  found together.  
The assumption of commuting limits is typically postulated
a priori and rigorous justification of this step is beyond the scope of this paper.

Let us denote the true
transmitted vector $\xvec_{0}$, so that $\yvec = \Hmat \xvec_{0} + \wvec$
is the generating model for the observation $\yvec$ and
we can equivalently write 
$\varV = \{\wvec, \xvec_{0}, \Hmat\}$.
Returning then to \textbf{Step~1}, we note that although the 
\emph{replicated} vectors $\{\xvec_{a}\}_{a=1}^{u}$
act as IID\ RVs drawn according to $p_{\varX}$ 
in \eqref{eq:integer_u_1} when conditioned on 
$\varV$, they can be correlated if not conditioned on $\varV$.
We examine this through the empirical correlations between  
the vectors in the set $\varX_{u+1} = \{\xvec_{a}\}_{a=0}^{u}$ using 
\emph{overlap matrix} 
$\mtxQ \in\mathbb{C}^{(u+1)\times(u+1)}$, whose $(a,b)$th element%
\footnote{The row/column indexes of $\mtxQ$ are $0,1,\ldots,u$ 
so that the correlations are measured also w.r.t.\ the true 
transmitted vector $\xvec_{0}$.  Furthermore, due to \eqref{eq:p_of_x}, 
the empirical correlations can be expected to converge 
to the true ones in the LSL postulated in \textbf{Step 2}.}
is given by $Q_{a,b} = M^{-1}\xvec^{\herm}_{b} \xvec^{}_{a}$.
Then, the structure that is imposed on $\mtxQ$ divides the replica analysis 
into two rough categories as described below.

\begin{step}[Replica symmetry]

	The \emph{RS ansatz} or \emph{RS assumption} means 
	that the indexes $a=1,\ldots,u$ are permutation symmetric and
	$\mtxQ$ can be written in terms of four parameters, for example,
	$Q_{0,0} = \Qp$, $Q_{0,a} = \Qm, a\geq 1,$
	$Q_{a,a} = \QQ, a\geq 1,$ and $Q_{a,b} = \Qq, a\neq b \geq 1$.  
	Note that $\mtxQ = \mtxQ^{\herm}$ by construction.
	If $\mtxQ$ is not of the RS form, it is said to have 
	replica symmetry breaking (RSB)
	structure whose  analysis is much more involved 
	\cite{Nishimori-2001, Mezard-Montanari-2009}. 

\end{step}

The importance of the RS assumption will become clear when we present a rough
sketch of the analysis of an ideal MIMO channel.  We also note that the 
overlap matrix given in \textbf{Step 3} allows the ``zeroth'' index to be 
treated separately to take into account the possibility
that either $\xvec_{0}$ has different distribution 
than $\xvec_{a}$ when $a \geq 1$, or the decoder uses mismatched statistics, i.e.,
$\funcZ(\varV,\xvec_{a})$ does not match the probability law
of the observation $\yvec = \Hmat \xvec_{0} + \wvec$ as in 
Appendix~\ref{sec:replicas}. 
For the simplified case considered below, however,
we have $\Qp = \QQ$ and $\Qm = \Qq$ since the 
indexes $a=0,1,\ldots,u$ can be treated on equal footing and two 
parameters is sufficient to define the RS form of $\mtxQ$.

Next we give a brief and informal example of 
replica analysis for ideal MIMO channel.  The reader may be surprised to find out that 
most of the discussion below deals with details about how to obtain the necessary 
formulas when we follow the three stages above and not about those stages per se.

\subsection{Average Over the Channel and Noise}

The starting point of our replica calculations 
is \eqref{eq:integer_u_1}, where we use the generating 
model of $\yvec$ to write in the 
exponential $\yvec - \Hmat \xvec_{a} = \wvec - \Hmat(\xvec_{0} - \xvec_{a})$.
The first task is then to compute the expectation w.r.t.\ $\wvec$ and $\Hmat$ 
for a fixed $\varX_{u+1} = \{\xvec_{a}\}_{a=0}^{u}$ that satisfies the correlations
of the RS overlap matrix $\mtxQ$.  Note that we cannot assume anymore that the vectors 
in $\varX_{u+1}$ are independent since we changed the order of expectations
in \eqref{eq:integer_u_1} and the average over $\varX_{u+1}$ is carried out 
(later) without conditioning on $\wvec$ and $\Hmat$.  With this in mind, 
it follows that given $\varX_{u+1}$, the set $\{\Hmat \xvec_{a}\}$ 
consists of CSCG RVs with  
correlations $\E_{\Hmat}\{(\Hmat \xvec_{a})(\Hmat \xvec_{b})^{\herm}\} = 
M^{-1}\xvec_{b}^{\herm}\xvec_{a}\Imat_{N} = Q_{a,b}\Imat_{N}$ that are deterministic
in the LSL.  Thus, we can replace $\{\Hmat(\xvec_{0} - \xvec_{a})\}_{a=1}^{u}$
by a set of CSCG RVs $\{\Dvec_{a}\}_{a=1}^{u}$ and use Gaussian integration
\eqref{eq:Gint} to average over both $\wvec$ and $\{\Dvec_{a}\}_{a=1}^{u}$ to obtain
(for details, see Appendix~\ref{sec:mismatched_avg_channel}.)
\begin{IEEEeqnarray}{rCl}
	\label{eq:Xi_Q_1}
\Xi(u) 
&=&
\int  \e^{N G^{(u)}(\mtxQ)} \mu(\mtxQ) \dx \mtxQ, \\
G^{(u)}(\mtxQ) 
&=& -u \ln \det [\Rmat_{\wvec}+(\QQ-\Qq)\Imat_{N}] \IEEEnonumber\\
&& \qquad \quad - u \ln \pi - \ln (u + 1) \IEEEeqnarraynumspace
\label{eq:Gu_1_simple}
\end{IEEEeqnarray}
where $\mtxQ$ should be understood to be in 
its RS parametrized form and $\mu(\mtxQ)$ is the PDF 
of the overlap matrix $\mtxQ$.

\begin{remark}
	\label{remark:u_and_universality}
	Firstly, note that due to the RS assumption (\textbf{Step~3}), 
	the function \eqref{eq:Gu_1_simple} is of a form that does not restrict
	$u$ to be an integer, as desired.  This is one of the reasons why we need 
	to express matrix $\mtxQ$ in a parametrized way instead of using it 
	``as-is''. 
	Secondly, there is some universality in this derivation and 
	the form \eqref{eq:Xi_Q_1} is a typical result of 
	replica analysis. 
	In some cases, however, different techniques are needed. 
	One example is non-IID\ ``mixing matrix'' that requires
	direct matrix integration \cite{Muller-Guo-Moustakas-2008,
	Vehkapera-Kabashima-Chatterjee-trit2013}.
\end{remark}

\subsection{Distribution of the Overlap Matrix and Large Deviations}

The second major step in the analysis is to find an
explicit formula for $\mu(\mtxQ)$, i.e., for
the probability weight of the set $\{ \xvec_{a} \}_{a=0}^{u}$
that satisfies $Q_{a,b} = M^{-1}\xvec_{b}^{\herm} \xvec_{a}^{}$.
The form
of \eqref{eq:Xi_Q_1} suggest that we should try to represent
$\mu(\mtxQ)$  as an exponential whose argument is linear in $N$ (or $M$)
so that we can employ Laplace's method or the 
method of steepest descent to evaluate
the integral w.r.t.\ $\mtxQ$.  If $\xvec_{a}\in \mathbb{R}^{M}$,
due to \eqref{eq:p_of_x}, the elements 
of $\xvec_{a}$ are IID\ for all $a = 0,1,\ldots,u$ and 
$\mu$ follows the large deviation principle \cite{Mezard-Montanari-2009,
Dembo-1998}.  Informally this implies%
\footnote{We use notation $a_{M} \asymp b_{M}$ to denote
``equality up to the leading exponential order'', that is
$\lim_{M\to\infty} M^{-1} \ln (a_{M} / b_{M}) = 0$.}
$\mu(\mtxQ) \asymp \e^{- M c^{(u)}(\mtxQ)}$, 
where the \emph{rate function} 
\begin{equation}
\label{eq:cfunc}
c^{(u)}(\mtxQ) = \sup_{\mtxQtil}
\bigg\{ \mathrm{tr} (\mtxQ \mtxQtil)
- \lim_{M\to\infty }\frac{1}{M} \ln \MGFu (\mtxQtil)
\bigg\},
\end{equation}
describes the exponential behavior of the probability, 
\begin{equation}
	\label{eq:MFGu}
\MGFu (\mtxQtil) =
\E_{X_{u+1}} \bigg\{
\exp\bigg(\sum_{a,b=0}^{u}
\tilde{Q}_{a,b} \xvec_{b}^{\herm} \xvec_{a}\bigg)\bigg\},
\end{equation}
is the moment generating function (MGF) associated with $\mu(\mtxQ)$ and
the supremum is over all $(u+1)\times(u+1)$
matrices $\mtxQtil$ that have the same RS form as 
$\mtxQ$, that is, 
$\tilde{Q}_{0,0} = \Qptil$, $\tilde{Q}_{0,a} = \Qmtil, a\geq 1,$
$\tilde{Q}_{a,a} = \QQtil, a\geq 1,$ and $\tilde{Q}_{a,b} = \Qqtil, a\neq b \geq 1$. 
Thus, we can assess \eqref{eq:Xi_Q_1} 
in the LSL up to the leading order by using the exponential form of $\mu$
and Laplace's method, namely,
\begin{IEEEeqnarray}{rCl}
	\label{eq:Xi_Q_2_1}
\Xi(u) 
&\asymp&
\int  \e^{M \alpha^{-1} G^{(u)}(\mtxQ)}\e^{- M c^{(u)}(\mtxQ)} \dx \mtxQ 
\IEEEnonumber\\
&=& \int  \exp\big(N [ \alpha^{-1} G^{(u)}(\mtxQ) -  c^{(u)}(\mtxQ)]\big) \dx \mtxQ
\IEEEeqnarraynumspace\\
&\asymp& \exp\bigg(M \sup_{\mtxQ,\mtxQtil} \big\{T^{(u)}(\mtxQ, \mtxQtil)\big\}\bigg),
	\label{eq:Xi_Q_2_2}
\end{IEEEeqnarray}
where we denoted for notational convenience
\begin{equation}
	T^{(u)}(\mtxQ, \mtxQtil) = 
\frac{1}{\alpha} G^{(u)}(\mtxQ) 
		- \mathrm{tr} (\mtxQ \mtxQtil)
		+ \lim_{M\to\infty }\frac{1}{M} \ln \MGFu (\mtxQtil).
\end{equation}
For complex vectors $\{\xvec_{a}\}$, 
the end result is essentially the same and the solution to 
the supremum is found among the critical points of the 
argument (see e.g., 
\cite{Wen-Wong-Chen-TCOM2007, Muller-Guo-Moustakas-2008, 
Takeuchi-Muller-Vehkapera-Tanaka-2013}).  
The large deviations analysis also guarantees that 
$\mtxQtil$ is in general a real symmetric matrix and 
if $(\mtxQ^{*},\mtxQtil^{*})$ is the solution
of the optimization problem in \eqref{eq:Xi_Q_2_2} 
then $T^{(u)}(\mtxQ^{*}, \mtxQtil^{*}) \in \mathbb{R}$, as expected
since $\outF$ is in our case real.

However, in RM there is some ambiguity as to whether the correct 
point in the saddle-point approximation \eqref{eq:Xi_Q_2_2} 
minimizes or maximizes the 
exponential when we let $u\to 0$ \cite{Nishimori-2001,Mezard-Montanari-2009}.
Thus, in RM, we seek in practice the critical points 
and \eqref{eq:F_alternative_1} is thus of the form
\begin{equation}
	\outF = -  \lim_{u\to 0^{+}}
		\frac{\partial}{\partial u} 
		\extr_{\mtxQ,\mtxQtil}
		\big\{T^{(u)}(\mtxQ, \mtxQtil)\big\},
\end{equation}
where $\extr_{X}\{h(X)\}$
denotes finding the critical points of a function $h(X)$.  

\subsection{Decoupled MGF and Critical Points}

The second part of replica analysis where the RS assumption
plays an important role (for the first one, see 
Remark~\ref{remark:u_and_universality}) is when we try to 
solve \eqref{eq:MFGu} and find the critical points of
$T^{(u)}(\mtxQ, \mtxQtil)$.  For the simplified setup in this section
where $\mtxQ$ and $\mtxQtil$ are represented with parameter 
$\{\QQ,\Qq\}$ and $\{\QQtil,\Qqtil\}$, respectively,
the MGF can be expressed as (see, e.g., \cite{Guo-Verdu-2005Jun} for details)
\begin{IEEEeqnarray}{rCl}
\MGFu(\mtxQtil) 
 =
 \prod_{m=1}^{M}\Bigg[ 
 \bigg(\frac{\Qqtil}{\pi}\bigg)^{-u}
 \!\!\! 
 \int \big[
 \E_{x_{m}}
 g(z_{m} \mid x_{m};\, \Qqtil^{-1})
 \big]^{u+1} \dx z_{m}\Bigg],
  \IEEEnonumber\\
\end{IEEEeqnarray}	
where $z_{m}$ are just dummy variables.
On the other hand, 
finding the critical points involves
taking eight partial derivatives for the RS case in \textbf{Step 3}
(for the simplified case here, four is enough).
Then, one should pick the solution that satisfies the conditions at the critical point
while providing the global extremum of \eqref{eq:F_alternative_1}.  
In the case considered here, we can actually get rid of two parameters since 
$\Qp = M^{-1}\E\|\xvec\|^{2}$ and $\Qptil = 0$ always at the critical point.
Note that if we did not parametrize $\mtxQ$, the critical points would be described 
by $u(u+1)$ equations and $\Xi(u)$ would depend explicitly on the fact that 
$u$ is an integer.  This is one of the reasons why even the full-RSB solution
(see \cite{Nishimori-2001,Mezard-Montanari-2009}) uses a round-about way of 
presenting $\mtxQ$ instead of using it ``as-is''.

Finally, we remark that it is quite common (see, e.g., \cite{Guo-Verdu-2005Jun}) 
to represent the end result in terms of new variables.  For example, if we have
equal transmit powers for each antennas $\gammabar = \gammabar_{m}$ 
in the simplified case considered here, then the parameters $\eta = \Qqtil$ and 
$\varepsilon = Q - q = \gammabar - q$ fully describe the RS matrices 
$\mtxQ$ and $\mtxQtil$ at the critical point.  
The former variable is inverse noise variance of a decoupled
Gaussian channel 
\begin{equation}
	z = x + w, \qquad p(w) = g(w \mid 0;\,\eta^{-1}), 
\end{equation}
and the latter variable $\varepsilon$ is the MMSE of this channel when the inputs are drawn
according to $p_{X}(x)$. 
The rest of RM is straightforward, albeit tedious algebra to arrive at 
\eqref{eq:F_alternative_1}.

\section{Replica Analysis for Mismatched Case}
\label{sec:replicas}

The analysis herein follows the main steps of RM as listed in 
Appendix~\ref{sec:replica_overview}.  Reader who is not 
familiar with the RM is encouraged to use discussion there as a 
guide to the derivations below.

\subsection{Replica Trick}

Let us consider the function $f(s)$ (free-energy) defined in
\eqref{eq:mismatched_free_energy}.  We then postulate that it can 
be expressed in the LSL using the standard 
replica trick (cf.\ Appendix~\ref{sec:replica_overview})
\begin{equation}
\label{eq:freeE_mismatched}
f(s) = 
 - 
 \lim_{M\to\infty} 
  \frac{1}{M}
 \lim_{u\to 0} \frac{\partial}{\partial u}
  \ln  \Xi^{(u,M)}(s),
\end{equation}
where we defined for later convenience
\begin{equation}
\Xi^{(u,M)}(s) = 
\E \bigg\{ \prod_{a=1}^{u} 
\e^{
-
[\wvec + \Hmat(\chivec_{0} - \chivec_{a})]^{\herm}
\RmatSP^{-1}
[\wvec + \Hmat(\chivec_{0} - \chivec_{a})]}\bigg\},
\label{eq:xi_mismatched}
\end{equation}
and denoted%
\footnote{
We remind the reader that for the case of mismatched decoding, 
the postulated covariance matrix $\RmatP$ is fixed 
by definition so that $\RmatSP=s^{-1}\RmatP$ is also a fixed predefined matrix. 
This is in contrast to the case of matched decoding 
\eqref{eq:true_decoding_pdf_expanded}, where the effective covariance 
matrix $\Rmat_{\wvec}+\Hmat \Rmat_{\vvec} \Hmat$ is random and depends
directly on the channel matrix $\Hmat$.
\label{fn:nonvalid_for_matched}}
$\RmatSP=s^{-1}\RmatP$ along with
$\chivec_{0} = \xvec_{0} + \vvec_{0}$
and $\chivec_{a} = \xvec_{a}, a = 1,\ldots,u$.
Here $\xvec_{0}$ is the original transmit vector in \eqref{eq:yvec_true}
and $\{\xvec_{a}\}_{a=1}^{u}$ are replicated 
data vectors, which are IID drawn according to $p(\xvec)$ when
conditioned on $\{\xvec_{0}, \vvec_{0}, \wvec, \Hmat\}$.  
On the other hand,  $\vvec_{0}$ represents the
noise plus distortion component at the transmit-side
that is CSCG with covariance matrix $\Rmat_{\vvec}$.
Starting with \eqref{eq:xi_mismatched}, 
the goal is then to obtain a functional expression for $\Xi^{(u,M)}(s)$
in the LSL that does not enforce $u$ to be an integer and then use 
\eqref{eq:freeE_mismatched} to obtain the desired quantity.
In the following, explicit limit notations are often 
omitted for notational convenience.

\subsection{Average Over the Channel and Noise}
\label{sec:mismatched_avg_channel}

To proceed with the evaluation of \eqref{eq:xi_mismatched}, 
we first make the RS assumption
\begin{IEEEeqnarray}{Cll}
\label{eq:RSsymm_p}
\Qp\, &= M^{-1} \|\chivec_{0}\|^{2}, \\
\Qm\; &= M^{-1} \chivec_{0}^{\herm} \chivec_{a}^{}, \qquad &a = 1,\ldots,u,\\
\QQ\; &= M^{-1} \|\chivec_{a}\|^{2}, \qquad &a = 1,\ldots,u, \\
\Qq\, &= M^{-1} \chivec_{a}^{\herm}\chivec_{b}^{}, \qquad &a\neq b \in
 \{1,\ldots,u\}. \IEEEeqnarraynumspace
 \label{eq:RSsymm_q}
\end{IEEEeqnarray}
and remind the reader that if we average first over $\Hmat$, 
the empirical correlations between $\{\xvec_{a}\}_{a=0}^{u}$ are not zero 
in general as discussed in Appendix~\ref{sec:replica_overview}.
Thus, noticing that
\begin{IEEEeqnarray}{l}
\E_{\Hmat}\{ [\Hmat(\chivec_{0} - \chivec_{a})^{}] 
[\Hmat(\chivec_{0} - \chivec_{b})]^{\herm}\}
\IEEEnonumber\\
\qquad =
\begin{cases}
\big[\Qp - (\Qm + \Qm^{*}) + \QQ\big]\Imat_{N},
\quad & a = b, \\
\big[\Qp - (\Qm + \Qm^{*}) + \Qq\big]\Imat_{N},
\quad & a \neq b,
\end{cases}
\IEEEeqnarraynumspace
	\label{eq:RS_correlations}
\end{IEEEeqnarray}
we may replace $\{\Hmat(\chivec_{0} - \chivec_{a})\}_{a=1}^{u}$ in 
\eqref{eq:xi_mismatched}
in the LSL by CSCG vectors
$\{\Dvec_a\}_{a=1}^{u}$ that are constructed as 
\begin{IEEEeqnarray}{rCl}
\Dvec_a &=& \tvec_{a}\sqrt{\QQ-\Qq} 
+ \uvec \sqrt{\Qp - (\Qm + \Qm^{*}) + \Qq} 
\IEEEeqnarraynumspace\\
 &=& \tvec_{a}\sqrt{A} + \uvec \sqrt{B},
 \label{eq:RS_v}
\end{IEEEeqnarray}
where $\big\{\uvec, \{\tvec_{a}\}_{a=1}^{u} \big\}$ are IID 
standard complex Gaussian RVs independent of $\wvec$. 
Plugging \eqref{eq:RS_v} into $\Xi^{(u,M)}(s)$
and recalling that $\RmatSP$ is a fixed
predefined matrix gives
\begin{IEEEeqnarray}{l}
\Xi^{(u,M)}(s)
=\frac{1}{\det(\Rmat_{\wvec})}
\E
\int 
\frac{\dx \wvec}{\pi^{N}}
\e^{-\wvec^{\herm}(\Rmat_{\wvec}^{-1}+u\RmatSP^{-1})\wvec}
\IEEEnonumber\\
\quad \times\int \frac{\dx \uvec}{\pi^{N}}
\e^{- \uvec^{\herm} (\Imat+u B\RmatSP^{-1})\uvec
- 2 \Re\{\wvec^{\herm} (u\sqrt{B}\RmatSP^{-1}) \uvec\}} 
\label{eq:eG_1}\\
\quad \times\bigg[\int 
\e^{- \tvec^{\herm}(
\Imat+A\RmatSP^{-1})\tvec
+ 2 \Re \{[- \sqrt{A}\RmatSP^{-1}( \wvec
+ \sqrt{B}\uvec)]^{\herm}\tvec\}} \frac{\dx \tvec}{\pi^{N}} \bigg]^{u}.   
\IEEEnonumber 
\end{IEEEeqnarray}
Next, Gaussian integration \eqref{eq:Gint}
is applied on the integral w.r.t.\ $\tvec$.  
Using also \eqref{eq:matrix_id1} 
we arrive at
\begin{IEEEeqnarray}{l}
\Xi^{(u,M)}(s) =
\E
\int 
\frac{\dx \wvec}{\pi^{N}}
\frac{\e^{-\wvec^{\herm} (\Rmat_{\wvec}^{-1}
+ u (A\Imat_{N}+\RmatSP)^{-1})\wvec}}{
[\det(\Imat+A\RmatSP^{-1})]^{u}
\det(\Rmat_{\wvec})} 
\IEEEeqnarraynumspace\\
\times\!\int\!
\e^{
- \uvec^{\herm}[
\Imat_{N}+u B (A\Imat_{N}+\RmatSP)^{-1}] \uvec
+ 2 \Re\{
[-u\sqrt{B}(A\Imat_{N}+\RmatSP)^{-1}
\wvec]^{\herm}\uvec\}} \frac{\dx \uvec}{\pi^{N}}.
\IEEEnonumber
\end{IEEEeqnarray}
Application of \eqref{eq:Gint} and \eqref{eq:matrix_id1} again 
for the integral w.r.t.\ $\uvec$ provides
\begin{IEEEeqnarray}{l}
\Xi^{(u,M)}(s)  \IEEEnonumber\\
= \E\Bigg\{\frac{\big[\det(\Imat+A\RmatSP^{-1})\big]^{-u}}{
\det\big[\Imat_{N}+u B (A\Imat_{N}+\RmatSP)^{-1}\big]
\det(\Rmat_{\wvec})} \IEEEnonumber\\
\quad\qquad \times
\int 
\e^{
-\wvec^{\herm}(
 \Rmat_{\wvec}^{-1} + u 
 [(A + u B) \Imat_{N}+ \RmatSP]^{-1})\wvec} \frac{\dx \wvec}{\pi^{N}} \Bigg\}\IEEEeqnarraynumspace\IEEEnonumber\\
=
\E\Bigg\{
\frac{\big(\det\big[
 \Imat_{N} + u 
 \Rmat_{\wvec}\big((A + u B) \Imat_{N}+ \RmatSP\big)^{-1}
\big]\big)^{-1}}{
\det\big[\Imat_{N}+u B (A\Imat_{N}+\RmatSP)^{-1}\big]
\big[\det(\Imat+A\RmatSP^{-1})\big]^{u}}\Bigg\},
\IEEEnonumber\\
\label{eq:after_GInt_detform}
\end{IEEEeqnarray}
where the second line is also obtained through Gaussian 
integration.
The above holds for any $\Rmat_{\wvec}$ and 
$\RmatSP$ that are Hermitian and invertible.
The determinants in \eqref{eq:after_GInt_detform} 
can be further simplified using \eqref{eq:det_id1},
so that recalling $\RmatSP = s^{-1}\RmatP$ and defining
two auxiliary matrices
\begin{IEEEeqnarray}{rCl}
	\label{eq:OmatAPP}
\Omat(\Qp,\Qm,\Qq) 
&=& \Rmat_{\wvec} 
+ (\Qp - (\Qm + \Qm^{*}) + \Qq)\Imat_{N}, \\
\OmatP(\QQ,\Qq)  &=& s^{-1}\RmatP+(\QQ-\Qq)\Imat_{N},
\label{eq:OmatPAPP}
\end{IEEEeqnarray}
that are both Hermitian, we finally have 
\begin{IEEEeqnarray}{rCl}
\Xi^{(u,M)}(s) &=&
\det(s^{-1}\RmatP)^{u}
\E \big\{ \e^{G^{(u)}(\Qp,\Qm,\Qq,\QQ)} \big\}, \\
G^{(u)}(\Qp,\Qm,\Qq,\QQ) &=&
(1-u)\ln \det\OmatP(\QQ,\Qq) \IEEEnonumber\\
\; && - \,
\ln\det\big[
 \OmatP(\QQ,\Qq) + u \Omat(\Qp,\Qm,\Qq)\big],
\IEEEeqnarraynumspace
\end{IEEEeqnarray}
Using the differentiation rule 
$\frac{\partial}{\partial x} \ln \det \vm{A}
=\tr \big(\vm{A}^{-1}
\frac{\partial\vm{A}}{\partial x}\big)$,
where the partial derivative should be understood
as an elementwise operation on $\vm{A}$,
we also obtain for later use the equalities
\begin{IEEEeqnarray}{rCl}
\label{eq:Gp}
\frac{\partial}{\partial \Qp}
G^{(u)}(\mtxQ) &=& %
-u  \tr \big( (\OmatP + u \Omat)^{-1}\big),
\\
\frac{\partial}{\partial \Qm}
G^{(u)}(\mtxQ) %
&=& \frac{\partial}{\partial \Qm^{*}}
G^{(u)}(\mtxQ) %
= u 
\tr \big( (\OmatP + u \Omat)^{-1}\big), 
\label{eq:Gm}\\
\frac{\partial}{\partial \Qq} G^{(u)}(\mtxQ) %
&=& u(u-1) 
\tr \big(\OmatP^{-1}
\Omat
(\OmatP + u \Omat)^{-1} \big), 
\label{eq:Gq}\\
  \frac{\partial}{\partial \QQ}G^{(u)}(\mtxQ) %
&=& u  \tr \big(\OmatP^{-1}
\Omat
(\OmatP + u \Omat)^{-1} \big)
- u \tr \big(\OmatP^{-1}\big), 
\label{eq:GQ}
\IEEEeqnarraynumspace
\end{IEEEeqnarray}
where the dependencies to 
$\{\Qp,\Qm,\Qq,\QQ\}$ were omitted on the RHSs
of the equations for notational simplicity.

\subsection{Distribution of the Overlap Matrix and Large Deviations}

Let us now 
write the general form of empirical correlations between 
$\{\Dvec_{a}\}$ as
\begin{IEEEeqnarray}{rCl}
\label{eq:Sab-element2}
\frac{1}{M}
\E_{\Hmat}\{\Dvec_{b}^{\herm} \Dvec_{a}\}
&=&
\bigg(
\frac{\|\chivec_{0}\|^{2}}{M}
- \frac{\chivec_{b}^{\herm} \chivec_{0}^{}}{M}
- \frac{\chivec_{0}^{\herm} \chivec_{a}^{}}{M}
+ \frac{\chivec_{b}^{\herm} \chivec_{a}^{}}{M}
\bigg) \IEEEnonumber\\
&=&
\big(
Q_{0,0}
-Q_{0,b} - Q_{a,0}
+ Q_{a,b}
\big),
\label{eq:S_as_Q}
\end{IEEEeqnarray}
where $Q_{a,b}$ are the elements of the overlap matrix
$\mtxQ \in \mathbb{C}^{(u+1)\times(u+1)}$ and 
have the obvious definitions.
We then need to find a suitable formula for the rate function \eqref{eq:cfunc}.
By the RS assumption,
\begin{equation}
\label{eq:trace_rs}
\mathrm{tr} (\mtxQ \mtxQtil)
= \Qp \Qptil + u\Qmtil(\Qm+\Qm^{*})
+u\QQ\QQtil + u (u-1) \Qq\Qqtil,
\end{equation}
since $\mtxQtil$ is real symmetric and we may write \eqref{eq:freeE_mismatched} as
in \eqref{eq:EZ_rs_1} at the top of the next page,
\begin{figure*}
\begin{IEEEeqnarray}{rCl}
 f_{\mathsf{RS}} &=&  
  -\lim_{M\to\infty} 
   \frac{1}{M}
  \ln \det(\RmatSP) -
\extr_{\mtxQ,\mtxQtil}
\lim_{M\to\infty}
\bigg\{
   \frac{1}{M}   \lim_{u\to 0} \frac{\partial}{\partial u}
   G^{(u)}(\mtxQ) \IEEEnonumber\\
&& \qquad \qquad  \qquad \qquad - 
  \lim_{u\to 0} \frac{\partial}{\partial u}\big[\Qp \Qptil + u(\Qm\Qmtil^{*}+\Qmtil\Qm^{*})
+u\QQ\QQtil + u (u-1) \Qq\Qqtil\big]
+\frac{1}{M} \sum_{m=1}^{M} \lim_{u\to 0} \frac{\partial}{\partial u} 
\ln \MGFu_{m} (\mtxQtil)
\bigg\} \IEEEeqnarraynumspace
\label{eq:EZ_rs_1}
\end{IEEEeqnarray}
\hrulefill
\end{figure*}
where the per-antenna rate function reads
\begin{equation}
\MGFu_{m} (\mtxQtil) = 
\E_{\{\chi_{a,m}\}} \Bigg\{\exp
\bigg[
\sum_{a=0}^{u}
\sum_{b=0}^{u}
\tilde{Q}_{a,b} \chi_{b,m}^{*} \chi_{a,m}
\bigg]\Bigg\},  
\label{eq:MGFm}
\end{equation}
and $\chivec_{a} = [\chi_{a,1} \; \cdots\;\, \chi_{a,M}]^{\trans}$.

\subsection{Decoupled MGF and Critical Points}
\label{sec:critical_points}

The first set of equations for the critical point arises from the equality
\begin{IEEEeqnarray}{l}
\frac{\partial}{\partial x}\mathrm{tr} (\mtxQ \mtxQtil)
= \frac{1}{M}\frac{\partial}{\partial x}G^{(u)} (\mtxQ), 
\end{IEEEeqnarray}
for $x\in\{\Qp,\Qm,\Qq,\QQ\}$.  The partial derivatives
on the LHS are trivial due to \eqref{eq:trace_rs} and the RHSs
we already obtained in 
\eqref{eq:Gp}--\eqref{eq:GQ}.  
If we drop the 
explicit dependence of $\Omat$ and $\OmatP$ on
$\{\Qp,\Qm,\Qq,\QQ\}$ for notational simplicity,
the RS conjugate parameters satisfy
\begin{IEEEeqnarray}{rCl}
\label{eq:Qptil}
\Qptil &=&
-u \frac{1}{M}
\tr \big[ (\OmatP + u \Omat)^{-1}\big] = -u \Qmtil,
\\
\label{eq:Qmtil}
\Qmtil &=&
\frac{1}{M}
\tr \big[ (\OmatP + u \Omat)^{-1}\big],
\\
\label{eq:Qqtil}
\Qqtil
&=&
\frac{1}{M}
\tr \big[\OmatP^{-1}
\Omat
(\OmatP + u \Omat)^{-1} \big],
\IEEEeqnarraynumspace \\
\QQtil  &=&
\frac{1}{M} \tr \big[\OmatP^{-1}\Omat
(\OmatP + u \Omat)^{-1} \big]
- \frac{1}{M} \tr \big(\OmatP^{-1}\big).
\IEEEeqnarraynumspace
\label{eq:QQtil}
\end{IEEEeqnarray}
Note that the above implies 
that in the limit $u\to 0$, we have
$\Qptil \to 0,$ and
$\Qmtil \to -(\QQtil - \Qqtil)$, so that the relevant critical
point can be written by using 
two instead of four ``tilde-parameters''.

The next task is to obtain an explicit
expression for the per-component moment generating function (MGF)
in \eqref{eq:MGFm} that does not require $u$
to be an integer.  Since this part is closely similar
to the analysis carried out, e.g., in \cite{Guo-Verdu-2005Jun}
we omit the details of the derivations.
Following the notation of \cite{Guo-Verdu-2005Jun}, we let 
$\xi = \Qmtil$ and 
$\eta = \Qmtil^{2} / \Qqtil$ which is sufficient to describe $\mtxQtil$
here.  Then, if we denote 
$\chi_{m} = x_{m} +  \vsym_{m}$ and
$\tilde{\chi}_{m} = \tilde{x}_{m}$, the scalar MGF
\eqref{eq:MGFm} can be written as
\begin{IEEEeqnarray}{ll}
	\label{eq:MGF_decoupled_2}
	\MGFu_{m} (\mtxQtil)  =& 
	\bigg(\frac{\pi}{\xi}\bigg)^{u}
	\E \bigg\{\int \dx z_{m}\,
	\e^{u \xi(|z_{m}|^{2} - |\chi_{m}|^{2})}
	\IEEEnonumber\\
	&\qquad \qquad \times
	p(z_{m} \mid \chi_{m})
	\big[\E_{\tilde{\chi}_{m}} q(z_{m} \mid \tilde{\chi}_{m})\big]^{u}
	\bigg\}, \IEEEeqnarraynumspace
	\label{eq:MGF_decoupled_3}
\end{IEEEeqnarray}
where
$p(z_{m} \mid \chi_{m}) =  g(z_{m} \mid \chi_{m};\,\eta^{-1})$ and
$q(z_{m} \mid \tilde{\chi}_{m}) 
= g(z_{m} \mid \tilde{\chi}_{m};\, \xi^{-1})$.
As a consequence of the above, 
$u$ does not need to be an integer anymore
and the limit $u\to 0$ is well defined.
From the partial derivatives of 
$\{\Qptil,\Qmtil,\Qqtil,\QQtil \}$ we obtain the 
second set of conditions at the critical point
\begin{IEEEeqnarray}{rCl}
\Qp &=& 
\lim_{M\to\infty }\frac{1}{M} 
\sum_{m=1}^{M}
\E |x_{m} + \vsym_{m}|^{2},
\label{Qp:final} \\
\QQ &=& 
\lim_{M\to\infty }\frac{1}{M} 
\sum_{m=1}^{M} 
\E \langle |\tilde{x}_{m}|^{2} \rangle_{q},
\IEEEeqnarraynumspace\\
\Qm &=& 
\lim_{M\to\infty }\frac{1}{M} 
\sum_{m=1}^{M} 
\E (x_{m} + \vsym_{m}) \langle \tilde{x}_{m}^{*}\rangle_{q},
\IEEEeqnarraynumspace\\
\Qq &=& 
\lim_{M\to\infty }\frac{1}{M} 
\sum_{m=1}^{M} 
\E\langle \tilde{x}_{m}^{*}\rangle_{q} \langle \tilde{x}_{m}^{*}\rangle_{q}, \IEEEeqnarraynumspace
\end{IEEEeqnarray}
where 
$x_{m},\tilde{x}_{m}\sim p(x_{m})$,
$\vsym_{m} \sim %
g(\vsym_{m} \mid 0;\, r^{m}_{\vvec})$,
\begin{IEEEeqnarray}{rCl}
\label{eq:decoupled_estimator}
\langle f(\tilde{x}_{m}) \rangle_{q} &=&  
\E_{\tilde{x}_{m}}
f(\tilde{x}_{m})
\frac{q(z_{m} \mid \tilde{x}_{m})}{q(z_{m})},
\IEEEeqnarraynumspace
\end{IEEEeqnarray}
and $q(z_{m}) = \E_{\tilde{x}_{m}} q(z_{m} \mid \tilde{x}_{m})$.
The interpretation is that
\eqref{eq:decoupled_estimator}
represents the conditional
mean estimator for postulated channel
$q(z_{m} \mid \tilde{\chi}_{m})$ when  
the true channel is given by $p(z_{m} \mid \chi_{m})$.
Then the true $\varepsilon = \Qp - (\Qm + \Qm^{*}) + \Qq$, 
and postulated $\tilde{\varepsilon} =  \QQ - \Qq$ MMSE 
reduce to 
\eqref{eq:mismatched_epsilon}~and~\eqref{eq:mismatched_epsilon_postulated}, respectively.
Finally, computing the partial derivatives w.r.t.\ $u$ in \eqref{eq:EZ_rs_1}
and taking the limit $u\to 0$ provides after some algebra
the free energy \eqref{eq:FreeE_mismatched_general}.

\bibliographystyle{IEEEtran}
\bibliography{TXnoiseMIMO}

\end{document}